%% file: paper.tex
\documentclass{sig-alternate}
\usepackage{outlines}
\usepackage{graphicx}
\usepackage{amssymb}
\usepackage{multirow}
\usepackage{subfigure}
\usepackage{amsmath, thm-restate}
\usepackage{algorithmic}
\usepackage[section]{algorithm}
\usepackage{subfigure}
\usepackage[normalem]{ulem}
\newskip\subfigcapskip	\subfigcapskip	= 2ex

\usepackage{url}

\toappear{}

\begin{document}


	\numberofauthors{2} 
	\author{
	Chao Li$^\dagger$, Michael Hay$^\dagger$, Vibhor Rastogi$^\ddagger$, Gerome Miklau$^\dagger$, Andrew McGregor$^\dagger$
	\and
	\alignauthor
	       \affaddr{$\dagger$University of Massachusetts Amherst,\\Amherst, Massachusetts, USA}\\
	       \affaddr{ \{chaoli,mhay,miklau,mcgregor\}@cs.umass.edu}
	\alignauthor
	       \affaddr{$\ddagger$University of Washington, \\Seattle, Washington, USA}\\
	       \affaddr{vibhor@cs.washington.edu}
	}

\twocolumn
\setlength{\columnsep}{.7cm}
	
\title{Optimizing Linear Counting Queries \\ Under Differential Privacy}

\maketitle{}

\pagestyle{empty}

\input{defs}

\paragraph*{N.B} This is the full version of the conference paper published as \cite{Li:2010Optimizing-Linear}.  This version includes an Appendix with proofs and additional results, and corrects a few typographical errors discovered after publication.  It also adds an improvement in the error bounds achieved under $(\epsilon,\delta)$-differential privacy, included as Theorem~\ref{thm:l2diffpriv}.

\begin{abstract}
Differential privacy is a robust privacy standard that has been successfully applied to a range of data analysis tasks.  But despite much recent work, optimal strategies for answering a collection of related queries are not known.

We propose the matrix mechanism, a new algorithm for answering a workload of predicate counting queries.  Given a workload, the mechanism requests answers to a different set of queries, called a query strategy, which are answered using the standard Laplace mechanism.  Noisy answers to the workload queries are then derived from the noisy answers to the strategy queries.  This two stage process can result in a more complex correlated noise distribution that preserves differential privacy but increases accuracy.

We provide a formal analysis of the error of query answers produced by the mechanism and investigate the problem of computing the optimal query strategy in support of a given workload.  We show this problem can be formulated as a rank-constrained semidefinite program.  Finally, we analyze two seemingly distinct techniques, whose similar behavior is explained by viewing them as instances of the matrix mechanism.
\end{abstract}

~\\ \noindent\textbf{Categories and Subject Descriptors:}
H.2.8 [\textbf{Database Management}]: Database Applications\textemdash
\textit{Statistical databases}; G.1 [\textbf{Numerical Analysis}]: Optimization
\vspace*{2.5mm}
~\\ \noindent\textbf{General Terms:}
Algorithms, Security, Theory
\vspace*{2.5mm}
~
\input{intro.tex}
\input{background.tex}

\input{query-answering.tex}



\input{decomposing.tex}  
\input{optimizing.tex}   
\input{applications.tex} 

\input{related-work.tex}
\input{conclusion.tex}
{ 
\bibliographystyle{abbrv}
\bibliography{paper}
}

\onecolumn
\appendix

\input{appendix.tex}

\end{document}

%% file: defs.tex
\floatname{algorithm}{Program}

\newcommand{\reals}{R}
\newcommand{\vol}{\textup{Vol}}
\newcommand{\convex}{\textup{Convex}}

\newcommand{\minerror}{\mbox{\sc MinError}}
\newcommand{\minsens}{\mbox{\sc MinSensitivity}}

\newcommand{\vect}[1]{\mathbf{#1}}
\newcommand{\sens}[1]{\Delta_{#1}}
\newcommand{\inv}[1]{{#1}^{-1}}
\newcommand{\ep}[1]{\inv{({#1}^t{#1})}}

\def\alg{\mathcal{K}}  
\def\LM{\mathcal{L}}	
\def\MM{\mathcal{M}}	

\def\tr{\mbox{trace}}
\def\var{\mbox{Var}}
\newcommand{\error}[2]{\mbox{\sc Error}_{#1}( #2 )}
\newcommand{\totalerror}[2]{\mbox{\sc TotalError}_{#1}( #2 )}
\newcommand{\maxerror}[2]{\mbox{\sc MaxError}_{#1}( #2 )}

\def\aa{\mathbb{A}}  
\def\bb{\mathbb{B}}  

\def\plus{{\!+}}
\def\b{\vect{\tilde{b}}}  
\def\x{\vect{x}}  
\def\estx{\vect{\hat x}}
\def\y{\vect{y}}
\def\q{\vect{q}}  
\def\w{\vect{w}} 
\def\v{\vect{v}}  
\def\estq{\vect{\hat q}}
\def\A{\vect{A}}
\def\B{\vect{B}}
\def\Q{\vect{Q}}
\def\W{\vect{W}}
\def\M{\vect{M}}
\def\D{\vect{D}}
\def\P{\vect{P}}
\def\I{\vect{I}}
\def\V{\vect{V}}
\def\H{\vect{H}}
\def\G{\vect{G}}
\def\R{\vect{R}}
\def\X{\vect{X}}
\def\Wav{\vect{Y}}

\def\PM{\P_{\M}}
\def\DM{\D_{\M}}
\def\DS{\D_s}
\def\DSinv{\DS^{-1}}


\def\WW{\W}		
\def\Wbool{\W_{01}}		
\def\Wrang{\W_{R}}		
\def\Wunit{\W_{unit}}		

\def\real{\mathbb{R}}

\def\RR{\vect{R}}

\newcommand{\ff}[1]{#1}


\newcommand{\mh}[1]{}
\newcommand{\gm}[1]{}
\newcommand{\eat}[1]{}
\newcommand{\cut}[1]{}

\newcommand{\set}[1]{\{#1\}}   

\newtheorem{definition}{Definition}[section]
\newtheorem{proposition}{Proposition}
\newtheorem{corollary}{Corollary}
\newtheorem{conjecture}{Conjecture}
\newtheorem{theorem}{Theorem}
\newtheorem{problem}{Problem}
\newtheorem{example}{Example}
\newtheorem{remark}{Remark}

\def\nbrs{nbrs}
\def\<{\langle}
\def\>{\rangle}

\def\qq{\tilde{q}}
\def\qbar{\overline{q}}

\def\Q{\mathbf{Q}}
\def\QQ{\mathbf{\tilde{Q}}}
\def\QC{\mathbf{\overline{Q}}}
\def\qq{\tilde{q}}
\def\qbar{\overline{q}}

\def\H{\mathbf{H}}
\def\HH{\mathbf{\tilde{H}}}
\def\HC{\mathbf{\overline{H}}}
\def\hh{\tilde{h}}
\def\hbar{\overline{h}}

\def\Lap{\mbox{Lap}}
\def\cnt{c}
\def\cons{\gamma}
\def\db{I}

\def\hght{{\ell}}  
\def\hv{\hght(v)}
\def\wt{\alpha}
\def\root{r}

\newcommand{\frob}[1]{||#1||_f}
\newcommand{\E}{\mathbb{E}}
\newcommand{\Ldist}[3]{||#1 -#2||_{#3}}
\newcommand{\rank}{\textup{rank}}
\newcommand{\trace}{\textup{Trace}}

\newcommand{\Ltwo}[1]{||#1||_2}
\newcommand{\Lone}[1]{\left\Vert #1  \right\Vert_1}

\newcommand{\reffull}[1]{#1}

\newtheorem{lemma}{Lemma}

\newcommand{\one}[1]{\mathbb{I}_{#1}}
\def\U{\mathcal U}
\def\Z{succZ} 
\def\s{s}
\def\m{M}
\def\mm{\tilde{M}}

%% file: intro.tex

\section{Introduction}

Differential privacy~\cite{dwork2006calibrating} offers participants in a dataset the compelling assurance that information released about the dataset is virtually indistinguishable whether or not their personal data is included.  It protects against powerful adversaries and offers precise accuracy guarantees.  As outlined in recent surveys~\cite{dwork2008differential,dwork2009differential,Dwork:2010A-firm-foundation}, it has been applied successfully to a range of data analysis tasks and to the release of summary statistics such as contingency tables~\cite{barak2007privacy}, histograms~\cite{Hay:2010Boosting-the-Accuracy,xiao2010differential}, and order statistics~\cite{nissim2007smooth}.  

Differential privacy is achieved by introducing randomness into query answers.  The original algorithm for achieving differential privacy, commonly called the Laplace mechanism~\cite{dwork2006calibrating}, returns the sum of the true answer and random noise drawn from a Laplace distribution.  The scale of the distribution is determined by a property of the query called its sensitivity: roughly the maximum possible change to the query answer induced by the addition or removal of one tuple.  Higher sensitivity queries are more revealing about individual tuples and must receive greater noise.

If an analyst requires only the answer to a single query about the database, then the Laplace mechanism has recently been shown optimal in a strong sense \cite{ghosh2009universally}.  But when multiple query answers are desired, an optimal mechanism is not known.  

At the heart of our investigation is the suboptimal behavior of the Laplace mechanism when answers to a set of correlated queries are requested.  We say two queries are correlated if the change of a tuple in the underlying database can affect both answers.  Asking correlated queries can lead to suboptimal results because correlation increases sensitivity and therefore the magnitude of the noise.  The most extreme example is when two duplicate queries are submitted.  The sensitivity of the pair of queries is twice that of an individual query.  This means the magnitude of the noise added to each query is doubled, but combining the two noisy answers (in the natural way, by averaging) gives a less accurate result than if only one query had been asked.  

Correlated workloads arise naturally in practice.  If multiple users are interacting with a database, the server may require that they share a common privacy budget to avoid the threat of a privacy breach from collusion.  Yet, in acting independently, they can easily issue redundant or correlated queries.  Further, in some settings it is appealing to simultaneously answer a large structured set of queries, (e.g. all range queries), which are inherently correlated.

In this work we propose the \emph{matrix mechanism}, an improved mechanism for answering a workload of predicate counting queries.  Each query is a linear combination of base counts reporting the number of tuples with the given combination of attribute values.  A set of such queries is represented as a matrix in which each row contains the coefficients of a linear query. Histograms, sets of marginals, and data cubes can be viewed as workloads of linear counting queries.  

The matrix mechanism is built on top of the Laplace mechanism.  Given a workload of queries, the matrix mechanism asks a different set of queries, called a {\em query strategy}, and obtains noisy answers by invoking the Laplace mechanism.  Noisy answers to the workload queries are then derived from the noisy answers to the strategy queries.  There may be more than one way to estimate a workload query from the answers to the strategy queries.  In this case the derived answer of the matrix mechanism combines the available evidence into a single consistent estimate that minimizes the variance of the noisy answer.

While the Laplace mechanism always adds independent noise to each query in the workload, the noise of the matrix mechanism may consist of a complex linear combination of independent noise samples.  Such correlated noise preserves differential privacy but can allow more accurate results, particularly for workloads with correlated queries.  

The accuracy of the matrix mechanism depends on the query strategy chosen to instantiate it.  This paper explores the problem of designing the optimal strategy for a given workload.  To understand the optimization problem we first analyze the error of any query supported by a strategy.  The error is determined by two essential features of the strategy: its \emph{error profile}, a matrix which governs the distribution of error across queries, and its \emph{sensitivity}, a scalar term that uniformly scales the error on all queries.  
%
%
Accurately answering a workload of queries requires choosing a strategy with a good error profile (relatively low error for the queries in the workload) and low sensitivity.  We show that natural strategies succeed at one, but not both, of these objectives.

We then formalize the optimization problem of finding the strategy that minimizes the total error on a workload of queries as a semi-definite program with rank constraints.  Such problems can be solved with iterative algorithms, but we are not aware of results that bound the number of iterations until convergence.  In addition, we propose two efficient approximations for deciding on a strategy, as well as a heuristic that can be used to improve an existing strategy.

Lastly, our framework encompasses several techniques proposed in the literature.  We use it to analyze two techniques~\cite{Hay:2010Boosting-the-Accuracy,xiao2010differential}, each of which can be seen as an instance of the matrix mechanism designed to support the workload consisting of all range queries.  Our analysis provides insight into the common behavior of these seemingly distinct techniques, and we prove novel bounds on their error.

After a background discussion we describe the matrix mechanism in Section~\ref{sec:query_answering}.  We analyze its error formally in Section~\ref{sec:analysis}.  In Section~\ref{sec:optimizing}, we characterize the optimization problem of choosing a query strategy and propose approximations.  We use our results to compare existing strategies in Section~\ref{sec:apps}.  We discuss related work, including other recent techniques that improve on the Laplace mechanism, in Section~\ref{sec:related}.

%% file: background.tex
\section{Background}  \label{sec:background}

This section describes the domain and queries considered, and reviews the basic principles of differential privacy.  We use standard terminology of linear algebra throughout the paper.  Matrices and vectors are indicated with bold letters (e.g $\A$ or $\x$) and their elements are indicated as $a_{ij}$ or $x_i$. For a matrix $\A$, $\A^t$ is its transpose, $\inv{\A}$ is its inverse, and $\tr(\A)$ is its trace (the sum of values on the main diagonal).  We use $diag(c_1, \dots c_n)$ to indicate an $n \times n$ diagonal matrix with scalars $c_i$ on the diagonal.  We use $\vect{0}^{m \times n}$ to indicate a matrix of zeroes with $m$ rows and $n$ columns.

\subsection{Linear queries}

The database is an instance $\db$ of relational schema $R(\aa)$, where $\mathbb{A}$ is a set of attributes.  We denote by $dom(\aa)$ the cross-product of the domains of attributes in $\aa$.  The analyst chooses a set of attributes $\bb \subseteq \aa$ relevant to their task.  For example if the analyst is interested in a subset of two dimensional range queries over attributes $A_1$ and $A_2$, they would set $\bb=\set{A_1,A_2}$.  We then form a frequency vector $\x$ with one entry for each element of $dom(\bb)$. For simplicity we assume $dom(\bb) = \{1,2,\ldots,n\}$ and for each $i \in dom(\bb)$, $x_i$ is the count of tuples equal to $i$ in the projection $\Pi_\bb(I)$. We represent $\x$ as a column vector of counts: $\x=[x_1 \dots x_n]^t$.

A {\em linear query} computes a linear combination of the counts in $\x$.  

\begin{definition}[Linear query]
A {\em linear query} is a length-$n$ row vector $\q=[q_1 \dots q_n]$ with each $q_i \in \mathbb{R}$.  
The answer to a linear query $\q$ on $\x$ is the vector product $\q\x = q_1x_1 + \dots + q_nx_n$.
\end{definition}

\eat{ 
\begin{example}
Our examples use a domain of size four: $dom=\set{1,2,3,4}$.  The linear query $[1, 1, 1, 1]$ returns the sum of all counts, i.e. the total number of tuples in the relation.  The linear query $[0,0,1,1]$ is a range query, returning the count of all tuples $3 \leq t.A \leq 4$.  More generally, the linear query $[2, 0, -1, 1.4]$ returns a linear combination of counts: $2x_1 - x_3 + 1.4x_4$.
\end{example}  }

We will consider sets of linear queries organized into the rows of a {\em query matrix}.  

\begin{definition}[Query matrix]
A {\em query matrix} is a collection of $m$ linear queries, arranged by rows to form an $m \times n$ matrix.
\end{definition}

If $\Q$ is an $m \times n$ query matrix, the query answer for $\Q$ is a length $m$ column vector of query results, which can be computed as the matrix product $\Q \x$.

\begin{example}
Figure \ref{fig:query-matrix} shows three query matrices, which we use as running examples throughout the paper.  $\I_4$ is the identity matrix of size four.  This matrix consists of four queries, each asking for an individual element of $\x$. $\H_4$ contains seven queries, which represent a binary hierarchy of sums: the first row is the sum over the entire domain (returning the total number of tuples in $\I$), the second and third rows each sum one half of the domain, and the last four rows return individual elements of $\x$.  $\Wav_4$ is the matrix of the Haar wavelet.  It can also be seen as a hierarchical set of queries: the first row is the total sum, the second row computes the difference between sums in two halves of the domain, and the last two rows return differences between smaller partitions of the domain.  In Section \ref{sec:apps} we study general forms of these matrices for domains of size $n$ \cite{Hay:2010Boosting-the-Accuracy,xiao2010differential}.  
\end{example}

\begin{figure}[t]
\centering 
 
\begin{tabular}{ccc}
\small $\begin{bmatrix} 
1 & 0 & 0 & 0\\
0 & 1 & 0 & 0\\
0 & 0 & 1 & 0 \\
0 & 0 & 0 & 1 \\
\end{bmatrix}$ 
& 
\small $\begin{bmatrix}
1 & 1 & 1 & 1 \\
1 & 1 & 0 & 0 \\
0 & 0 & 1 & 1 \\
1 & 0 & 0 & 0\\
0 & 1 & 0 & 0\\
0 & 0 & 1 & 0 \\
0 & 0 & 0 & 1 \\
\end{bmatrix}$ 
&  
 \small $\begin{bmatrix}
1 & 1 & 1 & 1\\
1 & 1 & \mbox{-}1 & \mbox{-}1\\
1 & \mbox{-}1 & 0 & 0 \\
0 & 0 & 1 & \mbox{-}1 \\
\end{bmatrix}$ \\
$\I_4$ & $\H_4$ & $\Wav_4$ \\
\end{tabular}
\caption{\label{fig:query-matrix} Query matrices with $dom=\set{1,2,3,4}$.  Each is full rank.  $I_4$ returns each unit count.  $H_4$ computes seven sums, hierarchically partitioning the domain.  $W_4$ is based on the Haar wavelet.}
\end{figure}

\subsection{The Laplace mechanism}

Because the true counts in $\x$ must be protected, only noisy answers to queries, satisfying differential privacy, are released.  We refer to the noisy answer to a query as an {\em estimate} for the true query answer.  The majority of our results concern classical $\epsilon$-diff\-erential privacy, reviewed below.  (We consider a relaxation of differential privacy briefly in Sec.~\ref{sec:sub:approx}.)

Informally, a randomized algorithm is differentially private if it produces statistically close outputs whether or not any one individual's record is present in the database.  For any input database $\db$, let $\nbrs(\db)$ denote the set of neighboring databases, each differing from $\db$ by at most one record; i.e., if $\db' \in \nbrs(\db)$, then $|(\db - \db') \cup (\db' - \db)| = 1$.
\eat{ 
\footnote{Differential privacy has been defined inconsistently in the literature.  The original concept, called $\epsilon$-indistinguishability~\cite{dwork2006calibrating}, defines neighboring databases using hamming distance rather than symmetric difference (i.e., $\db'$ is obtained from $\db$ by \emph{replacing} a tuple rather than adding/removing a tuple).  The choice of definition affects the calculation of query sensitivity.  We use the above definition (from Dwork~\cite{dwork2008differential}) but observe that our results also hold under indistinguishability, due to the fact that $\epsilon$-differential privacy (as defined above) implies $2\epsilon$-indistinguishability.}  
}

\begin{definition}[$\epsilon$-differential privacy] A randomized algorithm $\alg$ is $\epsilon$-differentially private if for any instance $I$, any $I' \in \nbrs(I)$, and any subset of outputs $S \subseteq Range(\alg)$, the following holds:
\[
Pr[ \alg(I) \in S] \leq \exp(\epsilon) \times Pr[ \alg(I') \in S],
\]		
where the probability is taken over the randomness of the $\alg$.
	\end{definition}

Differential privacy can be achieved by adding random noise to query answers.  The noise added is a function of the privacy parameter, $\epsilon$, and a property of the queries called {\em sensitivity}.  The sensitivity of a query bounds the possible change in the query answer over any two neighboring databases.  For a single linear query, the sensitivity bounds the absolute difference of the query answers.  For a query matrix, which returns a vector of answers, the sensitivity bounds the $L_1$ distance between the answer vectors resulting from any two neighboring databases.   
The following proposition extends the standard notion of query sensitivity to query matrices.  Note that because two neighboring databases $\db$ and $\db'$ differ in exactly one tuple, it follows that their corresponding vectors $\x$ and $\x'$ differ in exactly one component, by exactly one. 

\begin{proposition}[Query matrix sensitivity]
The sensitivity of matrix $\Q$, denoted $\sens{\Q}$, is:  
\begin{eqnarray*}
\sens{\Q} &=^{def}& \max_{\Lone{\x-\x'}=1} \Lone{\Q\x - \Q\x'} 
			= \max_{j}\sum_{i=1}^n |q_{ij}|.
\end{eqnarray*}
Thus the sensitivity of a query matrix is the maximum $L_1$ norm of a column.
\end{proposition}

\begin{example}
The sensitivities of the query matrices in Figure \ref{fig:query-matrix} are: $\sens{\I_4}=1$ and $\sens{\H_4}=\sens{\Wav_4}=3$.  A change by one in any component $\x_i$ will change the query answer $\I_4\x$ by exactly one, but will change $\H_4\x$ and $\Wav_4\x$ by three since each $x_i$ contributes to three linear queries in both $\H_4$ and $\Wav_4$.
\end{example}

The following proposition describes an $\epsilon$-differ\-entially private algorithm, adapted from Dwork et al.~\cite{dwork2008differential}, for releasing noisy answers to the workload of queries in matrix $\W$.  The algorithm adds independent random samples from a scaled Laplace distribution.

\begin{proposition}[Laplace mechanism] \label{prop:laplace}
Let $\W$ be a qu\-ery matrix consisting of $m$ queries, and let $\b$ be a length-$m$ column vector consisting of independent samples from a Laplace distribution with scale $1$. Then the randomized algorithm $\LM$ that outputs the following vector is $\epsilon$-differentially private:
$$\LM(\W,\x) = \W\x + (\frac{\sens{\W}}{\epsilon})\b.$$
\end{proposition}

Recall that $\W\x$ is a length-$m$ column vector representing the true answer to each linear query in $\W$.  The algorithm adds independent random noise, scaled by $\epsilon$ and the sensitivity of $\W$.  Thus $\LM(\W,\x)$, which we call the {\em output vector}, is a length-$m$ column vector containing a noisy answer for each linear query in $\W$.

%% file: query-answering.tex
\section{The Matrix Mechanism}
\label{sec:query_answering}

Central to our approach is the distinction between a query {\em strategy} and a query {\em workload}.  Both are sets of linear queries represented as matrices.  The workload queries are those queries for which the analyst requires answers.  Submitting the workload queries to the Laplace mechanism described above is the standard approach, but may lead to greater error than necessary in query estimates.  Instead we submit a different set of queries to the differentially private server, called the query strategy.  We then use the estimates to the strategy queries to derive estimates to the workload queries.  Because there may be more than one derived estimate for a workload query, we wish to find a single consistent estimate with least error. 

In this section we present the formal basis for this derivation process. We define the set of queries whose estimates can be derived and we provide optimal mechanisms for deriving estimates.  Using this derivation, we define the {\em matrix mechanism}, an extension of the Laplace mechanism that uses a query strategy $\A$ to answer a workload $\W$ of queries.  The remainder of the paper will then investigate, given $\W$, how to choose the strategy $\A$ to instantiate the mechanism.

\eat{  
For example, workloads we will consider include $\Wrang$, the set of all range queries over $dom(\bb)$, and $\Wbool$, the set of all predicate queries (with coefficients of 0 or 1).  These workloads are large and have high sensitivity: answering each query in the workload would result in very noisy results for all queries.
Our goal is to devise optimal query strategies that support these workloads, as well as other ad hoc workloads that may be relevant to particular analysis tasks. 
}

\subsection{Deriving new query answers}

Suppose we use the Laplace mechanism to get noisy answers to a query strategy $\A$.  Then there is sufficient evidence, in the noisy answers to $\A$, to construct an estimate for a workload query $\w$ if $\w$ can be expressed as a linear combination of the strategy queries:
\begin{definition}[Queries supported by a strategy] \label{def:sup}
A strategy $\A$ supports a query $\w$ if $\w$ can be expressed as a linear combination of the rows of $\A$.
\end{definition}
In other words, $\A$ supports any query $\w$ that is in the subspace defined by the rows of $\A$.  If a strategy matrix consists of at least $n$ linearly independent row vectors (i.e., its row space has dimension $n$), then it follows immediately that it supports all linear queries.  Such matrices are said to have {\em full rank}.  We restrict our attention to full rank strategies and defer discussion of this choice to the end of the section.


To derive new query answers from the answers to $\A$ we first compute an estimate, denoted $\estx_\A$, of the true counts $\x$.  Then the derived estimate for an arbitrary linear query $\w$ is simply the vector product $\w\estx_\A$.  The estimate of the true counts is computed as follows:

\begin{definition}[Estimate of $\x$ using $\A$] \label{def:estimate}
Let $\A$ be a full rank query strategy $\A$ consisting of $m$ queries, and let $\y = \LM(\A,\x)$ be the noisy answers to $\A$.  Then $\estx_\A$ is the estimate for $\x$ defined as:
$$\estx_\A = \A^\plus\y,$$
where $\A^\plus = \inv{(\A^t\A)}\A^t$ is the {\em pseudo-inverse} of $\A$. 
\end{definition}

Because $\A$ has full rank, the number of queries in $\A$, $m$, must be at least $n$.  When $m=n$, then $\A$ is invertible and $\A^\plus=\inv{\A}$. Otherwise, when $m>n$, $\A$ is not invertible, but $\A^\plus$ acts as a left-inverse for $\A$ because $\A^{+}\A=\I$. We explain next the justification for the estimate $\estx_\A$ above, and provide examples, considering separately the case where $m=n$ and the case where $m>n$.  

~ \\ \noindent \textbf{$\A$ is square. } 
In this case $\A$ is an $n \times n$ matrix of rank $n$, and it is therefore invertible.  Then given the output vector $\y$, it is always possible to compute a unique estimate for the true counts by inverting $\A$.  The expression in Definition \ref{def:estimate} then simplifies to : 
 	$$\estx_\A = \inv{\A}\y.$$
	In this case, query strategy $\A$ can be viewed as a linear transformation of the true counts, to which noise is added by the privacy mechanism.  The transformation is then reversed, by the inverse of $\A$, to produce a consistent estimate of the true counts.  
	
\begin{example}
In Figure \ref{fig:query-matrix}, $\I_4$ and $\Wav_4$ are both square, full rank matrices which we will use as example query strategies.  The inverse of $\I_4$ is just $\I_4$ itself, reflecting the fact that since $\I_4$ asks for individual counts of $\x$, the estimate $\estx$ is just the output vector $\y$.  The inverse of $\Wav_4$ is shown in Figure \ref{fig:estimate}(c).  Row $i$ contains the coefficients used to construct an estimate of count $x_i$.  For example, the first component of $\estx_\A$ will be computed as the following weighted sum: $.25y_1 + .25y_2 + .5y_3$.
\end{example}

Specific transformations of this kind have been studied before.  A Fourier transformation is used in \cite{barak2007privacy}, however, rather than recover the entire set of counts, the emphasis is on a set of marginals.  A transformation using the Haar wavelet is considered \cite{xiao2010differential}.  Our insight is that any full rank matrix is a viable strategy, and our goal is to understand the properties of matrices that make them good strategies.  In Section \ref{sec:apps} we analyze the wavelet technique \cite{xiao2010differential} in detail.

~ \\ \noindent \textbf{$\A$ is rectangular.} When $m>n$, we cannot invert $\A$ and we must employ a different technique for deriving estimates for the counts in $\x$.  In this case, the matrix $\A$ contains $n$ linearly independent rows, but has additional row queries as well.  These are additional noisy observations that should be integrated into our estimate $\estx_\A$.  Viewed another way, we have a system of equations given by $\y=\A\x$, with more equations ($m$) than the number of unknowns in $\x$ ($n$).  The system of equations is likely to be inconsistent due to the addition of random noise.  

We adapt techniques of linear regression, computing an estimate $\estx_\A$ that minimizes the sum of the squared deviations from the output vector.  Because we assume $\A$ has full rank, this estimate, called the {\em least squares} solution, is unique.  The expression in Definition \ref{def:estimate} computes the well-known least squares solution as $\estx = \inv{(\A^t\A)}\A^t\y$.

This least squares approach was originally proposed in \cite{Hay:2010Boosting-the-Accuracy} as a method for avoiding inconsistent answers in differentially private outputs, and it was shown to improve the accuracy of a set of histogram queries.  In that work, a specific query strategy is considered (related to our example $\H_4$) consisting of a hierarchical set of queries.  An efficient algorithm is proposed for computing the least squares solution in this special case.  We analyze this strategy further in Sec.~\ref{sec:apps}.

\begin{example}
$\H_4$, shown in Figure \ref{fig:query-matrix}, is a rectangular full rank matrix with $m=7$.  The output vector $\y=\LM(\H_4,\x)$ does not necessarily imply a unique estimate.  For example, each of the following are possible estimates of $x_1$: $y_4$, $y_2-y_5$, $y_1-y_3-y_5$, each likely to result in different answers.  The reconstruction matrix for $\H_4$, 
$\H_4^\plus$ shown in Fig~\ref{fig:estimate}(b), describes the unique least squares solution.  The estimate for $x_1$ is a weighted combination of values in the output vector: $\frac{3}{21}y_1+\frac{5}{21}y_2-\frac{2}{21}y_3+\frac{13}{21}y_4-\frac{8}{21}y_5-\frac{1}{21}y_6-\frac{1}{21}y_7$.  Notice that greatest weight is given to $y_4$, which is the noisy answer to the query that asks directly for $x_1$; but the other output values contribute to the final estimate.
\end{example}


In summary, whether $m=n$ or $m>n$, Definition \ref{def:estimate} shows how to derive a unique, consistent estimate $\estx_\A$ for the true counts $\x$.  Once $\estx_\A$ is computed, the estimate for any $\w$ is computed as $\w\estx_\A$.  
The following theorem shows that $\estx_\A$ is an unbiased estimate of $\x$ and that in a certain sense it is the best possible estimate given the answers to strategy query $\A$. 

\begin{theorem}[Minimal Variance of estimate of $\x$] 
	\label{thm:estimator}
	Given noisy output $\y=\LM(\A,\x)$, the estimate $\estx = \A^\plus\y$ is unbiased (i.e., $\E [ \estx_\A ] = \x$), and has the minimum variance among all unbiased estimates that are linear in $\y$. 	
\end{theorem}

The theorem follows from an application of the Gauss-Markov theorem~\cite{gaussmarkov} and extends a similar result from \cite{Hay:2010Boosting-the-Accuracy}. 


\begin{figure*}[t]
\centering
\subfigure[$\inv{\I_4}$]{
$
\begin{bmatrix}
1 & 0 & 0 & 0\\
0 & 1 & 0 & 0\\
0 & 0 & 1 & 0 \\
0 & 0 & 0 & 1 \\
\end{bmatrix}
$
} \hfill
\subfigure[$\H_4^\plus$]{
$\frac{1}{21} \times
\begin{bmatrix} 
    \ff{3} & \ff{5}  &  \ff{-2} &  \ff{13} & \ff{-8} & \ff{-1} & \ff{-1} \\
    \ff{3} & \ff{5}  &  \ff{-2} & \ff{-8}  & \ff{13} & \ff{-1} & \ff{-1} \\ 
    \ff{3} & \ff{-2} &  \ff{5}  &  \ff{-1} & \ff{-1} & \ff{13} & \ff{-8} \\ 
    \ff{3} & \ff{-2} &  \ff{5}  &  \ff{-1} & \ff{-1} & \ff{-8} & \ff{13} \\
\end{bmatrix}
$
} \hfill
\subfigure[$\inv{\Wav_4}$]
{$
\begin{bmatrix}
 0.25 & 0.25 & 0.5 & 0.0 \\
 0.25 & 0.25 &-0.5 & 0.0 \\ 
 0.25 &-0.25 & 0.0 & 0.5 \\
 0.25 &-0.25 & 0.0 &-0.5 \\
\end{bmatrix}
$ }
\caption{\label{fig:estimate} For strategy $\A$ equal to $\I_4$, $\H_4$ and $\Wav_4$, respectively, the matrices above are used to derive the estimate $\estx_\A$ from the noisy output $\y=\LM(\A,\x)$.  Row $i$ in each matrix contains the coefficients of the linear combination of $\y$ used to construct an estimate for count $\x_i$. The inverse of the identity $\I_4$ is the identity; the reconstruction matrix for $\H_4$ is $\H_4^\plus=\inv{(\H_4^t\H_4)}\H_4^t$; $\inv{\Wav_4}$ describes the wavelet reconstruction coefficients.}
\end{figure*}

\subsection{The Matrix Mechanism}
In the presentation above, we used the Laplace mechanism to get noisy answers $\y$ to the queries in $\A$, and then derived $\estx_\A$, from which any workload query could then be estimated.  It is convenient to view this technique as a new differentially private mechanism which produces noisy answers to workload $\W$ directly.  This mechanism is denoted $\MM_\A$ when instantiated with strategy matrix $\A$.

\begin{proposition}[Matrix Mechanism] \label{def:m-mech} Let $\A$ be a full rank $m \times n$ strategy matrix, let $\W$ be any $p \times n$ workload matrix, and let $\b$ be a length-$m$ column vector consisting of independent samples from a Laplace distribution with scale $1$. Then the randomized algorithm $\MM_\A$ that outputs the following vector is $\epsilon$-differentially private:
\begin{eqnarray*}
\MM_\A(\W,\x) &=& \W\x + (\frac{\sens{\A}}{\epsilon})\W \A^\plus \b.
\end{eqnarray*}
\end{proposition}
\begin{proof}The expression above can be rewritten as follows:
\begin{align*}
\MM_\A(\W,\x) &=\W(\x+(\frac{\sens{\A}}{\epsilon}) \A^\plus \b)\\
&=\W\A^\plus(\A\x+(\frac{\sens{\A}}{\epsilon})\b) \\
&=\W\A^\plus\LM(\A,\x).
\end{align*}
Thus, $\MM_\A(\W,\x)$ is simply a post-processing of the output of the $\epsilon$-differentially private $\LM$ and therefore $\MM$ is also $\epsilon$-differentially private.
\end{proof}

Like the Laplace mechanism, the matrix mechanism computes the true answer, $\W\x$, and adds to it a noise vector.  But in the matrix mechanism the independent Laplace noise $\b$ is transformed by the matrix $\W \A^\plus$, and then scaled by $\sens{\A}/\epsilon$.  The potential power of the mechanism arises from precisely these two features.  First, the scaling is proportional to the sensitivity of $\A$ instead of the sensitivity of $\W$, and the former may be lower for carefully chosen $\A$.  Second, because the noise vector $\b$ consists of independent samples, the Laplace mechanism adds independent noise to each query answer.  However, in the matrix mechanism, the noise vector $\b$ is transformed by $\W\A^\plus$.  The resulting noise vector is a linear combination of the independent samples from $\b$, and thus it is possible to add {\em correlated} noise to query answers, which can result in more accurate answers for query workloads whose queries are correlated.

\eat{ 
 The accuracy achieved by the matrix mechanism depends on the choice of $\A$, however if the workload is full rank, using $\W$ as the strategy in the matrix mechanism improves over 
}

\subsection{Rank deficient strategies and workloads}
If a workload is not full rank, then it follows from Definition \ref{def:sup} that a full rank strategy is not needed.  Instead, any strategy whose rowspace spans the workload queries will suffice.  However, if we wish to consider a rank deficient strategy $\B$, we can always transform $\B$ into a full rank strategy $\A$ by adding a scaled identity matrix $\delta\I$, where $\delta$ approaches zero.  The result is a full rank matrix $\A$ which supports all queries, but for which the error for all queries not supported by $\B$ will be extremely high.  Another alternative is to apply dimension reduction techniques, such as principle components analysis, to the workload queries to derive a transformation of the domain in which the workload has full rank.  Then the choice of full rank strategy matrix can be carried out in the reduced domain.

\eat{  
\begin{remark}[Uniform noise] 
Both mechanisms above have been formulated using a noise vector $b$ which consists of samples with uniform variance in each component.  This formulation simplifies notation, but we note that it does not limit the generality of the mechanisms.  In Appendix \ref{app:unequal_noise} we show that the addition of non-uniform noise can be simulated by linearly scaling the queries in a query matrix.
\end{remark}
}

%% file: decomposing.tex
\section{The analysis of error} \label{sec:analysis}

In this section we analyze the error of the matrix mechanism formally, providing closed-form expressions for the mean squared error of query estimates.  We then use matrix decomposition to reveal the properties of the strategy that determine error.  This analysis is the foundation for the optimization problems we address in the following section.

\eat{Ultimately, our goal is to understand the error function for existing strategies, and to design new strategies that are tailored to a specific workload.  In this section 
we show that the error function defines an elliptic paraboloid over the $n$-dimensional query vector space.  A query vector defines a point in $\real^n$ and the error for that query corresponds to the height of the paraboloid at that point.
The shape of this paraboloid depends on the query strategy $\A$ and 
we use the singular value decomposition (SVD) to factor a query strategy $\A$ into a product of three matrices, where each matrix determines one aspect of the paraboloid's shape.
One matrix scales it, another stretches it, and the third orients it.  The last two matrices define the error profile of a strategy.  The first matrix, which scales the height of the paraboloid, is related to the sensitivity of a strategy.  We show that it is possible to change a query strategy---by rotating its columns---such that its error profile is preserved but its sensitivity may be lowered.}

\subsection{The error of query estimates}

While a full rank query strategy $\A$ can be used to compute an estimate for {\em any} linear query $\w$, the accuracy of the estimate varies based on the relationship between $\w$ and $\A$. We measure the error of strategy $\A$ on query $\w$ using mean squared error.

\begin{definition}[Query and Workload Error] Let $\estx_\A$ be the estimate for $\x$ derived from query strategy $\A$.  The mean squared error of the estimate for $\w$ using strategy $\A$ is: $$\error{\A}{\w} = \E[ ( \w\x - \w\estx_\A )^2 ].$$ 
Given a workload $\W$, the total mean squared error of answering $\W$ using strategy $\A$ is: $$\totalerror{\A}{\W} = \sum_{\w_i \in \W} \error{\A}{\w_i}.$$	
\end{definition}

%
%

For any query strategy $\A$ the following proposition describes how to compute the error for any linear query $\w$ and the total error for any workload $\W$:

\begin{proposition}[Error Under Strategy $\A$]
  \label{prop:error:single}
  \label{prop:error:workload}
For a full rank query matrix $\A$ and linear query $\w$, the estimate of $\w$ is unbiased (i.e. $\E[ \w\estx_\A ] = \w\x$), and the error of the estimate of $\w$ using $\A$ is equal to:
	\begin{equation} \label{eq:var}
		\error{\A}{\w}
		= (\frac{2}{\epsilon^2}) \; \sens{\A}^2 \; \w \inv{(\A^t\A)} \w^t.
	\end{equation}
The total error of the estimates of workload $\W$ using $\A$ is:
	\begin{eqnarray} \label{eq:toterr}
	\totalerror{\A}{\W} = (\frac{2}{\epsilon^2}) \; \sens{\A}^2 \; \tr (\inv{(\A^t\A)} \W^t\W).
		\end{eqnarray}
\end{proposition}

\begin{proof}
	It is unbiased because $\estx_\A$ is unbiased.  Thus, for formula (\ref{eq:var}), the mean squared error is equal to the variance:
	\begin{align*}
		\error{\A}{\w} &=  Var( \w \estx_\A ) 
		= Var( \w \x + (\frac{\sens{\A}}{\epsilon}) \w \A^\plus \b) \\
		&= (\frac{\sens{\A}}{\epsilon})^2 Var( \w \A^\plus \b ).
	\end{align*}
With algebraic manipulation and that $Var(\b) = 2\I_m$, we get:
	\begin{align*}
		Var( \w \A^\plus \b ) &= \w \A^\plus Var(\b) (\w \A^\plus)^t \\
		&= \w \A^\plus 2 \I_m (\w \A^\plus)^t \\
		&= 2 \w (\A^t \A)^{-1} \A^t \A ((\A^t \A)^{-1})^t \w^t \\
		&= 2 \w (\A^t \A)^{-1} \w^t,
	\end{align*}
where $((\A^t \A)^{-1})^t = (\A^t \A)^{-1}$ because the matrix is symmetric.  Therefore, 
$\error{\A}{\w} = (\frac{\sens{\A}}{\epsilon})^2 2 \w (\A^t \A)^{-1} \w^t$.

For formula (\ref{eq:toterr}) if $\w_i$ is row $i$ of workload $\W$, then $\error{\A}{\w_i}$ is the $i$-th entry on the diagonal of matrix $(\frac{2}{\epsilon^2})\sens{\A}^2\W \inv{(\A^t\A)} \W^t$.  Therefore, since the trace of a matrix is the sum of the values on its diagonal, the $\totalerror{\A}{\W}$ is equal to $$(\frac{2}{\epsilon^2})\sens{\A}^2\tr(\W \inv{(\A^t\A)} \W^t).$$
Formula \ref{eq:toterr} follows from a standard property of the trace: $\tr(\W \inv{(\A^t\A)} \W^t)=\tr( \inv{(\A^t\A)} \W^t\W)$.
\end{proof}

These formulas underlie much of the remaining discussion in the paper.  Formula (\ref{eq:var}) shows that, for a fixed $\epsilon$, error is determined by two properties of the strategy: (i) its squared sensitivity, $\sens{\A}^2$; and (ii) the term $\w \inv{(\A^t\A)} \w^t$.  In the sequel, we refer to the former as simply the {\em sensitivity term}.  We refer to the latter term as the {\em profile term} and we call matrix $\ep{\A}$ the {\em error profile} of query strategy $\A$.  

\begin{definition}[Error profile]
For any full rank $m \times n$ query matrix $\A$, the error profile of $\A$, denoted $\M$, is defined to be the $n \times n$ matrix $\ep{\A}$.
\end{definition}

The coefficients of the error profile $\M=\ep{\A}$ measure the covariance of terms in the estimate $\estx_\A$.  Element $m_{ii}$ measures the variance of the estimate of $x_i$ in $\estx_\A$ (and is always positive), while $m_{ij}$ measures the covariance of the estimates of $x_i$ and $x_j$ (and may be negative).  \eat{Since the estimate of query $\w$ is the linear combination of $\estx_\A$ described by the product $\w\estx_\A$, the error of that estimate is proportional to the profile term $\w\ep{\A}\w^t$ in Formula (\ref{eq:var}).}
We can equivalently express the profile term as: $$\w\ep{\A}\w^t = \sum_i w_i^2 m_{ii} + \sum_{i < j} 2w_i w_j m_{ij},$$
which shows that error is a weighted sum of the (positive) diagonal variance terms of $\M$, plus a (possibly negative) linear combination of off-diagonal covariance terms.  This illustrates that it is possible to have a strategy that has relatively high error on individual counts yet is quite accurate for other queries that are linear combinations of the individual counts.  
We analyze instances of such strategies in Sec.~\ref{sec:apps}.

\begin{figure*}[t]
\centering
\subfigure[$\ep{\I_4}$]{
$ 
\begin{bmatrix}
1 & 0 & 0 & 0\\
0 & 1 & 0 & 0\\
0 & 0 & 1 & 0 \\
0 & 0 & 0 & 1 \\
\end{bmatrix}$ \vspace{3ex}
}
\hfill
\subfigure[$\ep{\H_4}$]{
$ \frac{1}{21} \times
\begin{bmatrix} 
    \ff{13} & \ff{-8}  &  \ff{-1} &  \ff{-1}  \\
    \ff{-8} & \ff{13}  &  \ff{-1} & \ff{-1}   \\ 
    \ff{-1} & \ff{-1} &  \ff{13}  &  \ff{-8}  \\ 
    \ff{-1} & \ff{-1} &  \ff{-8}  &  \ff{13}  \\
\end{bmatrix}
$
}
\hfill
\subfigure[$\ep{\Wav_4}$]{
$  \frac{1}{8} \times
\begin{bmatrix}
 3 & -1 & 0 & 0 \\
 -1 & 3 & 0 & 0 \\ 
 0 & 0 & 3 & -1 \\
 0 & 0 & -1 & 3 \\
\end{bmatrix}
$ 
}
\caption{\label{fig:error} The profile term of the error function for query $\w$ on strategy $\A$ is $\w\ep{\A}\w^t$.  Shown are the error profiles for query strategies $\I_4$, $\H_4$, and $\Wav_4$. Every error profile matrix is symmetric and positive definite.}
\end{figure*}

\begin{example} Figure \ref{fig:error} shows the error profiles for each sample strategy. $\I_4$ has the lowest error for queries that ask for a single count of $\x$, such as $\w=[1,0,0,0]$.  For such queries the error is determined by the diagonal of the error profile (subject to scaling by the sensitivity term). Queries that involve more than one count will sum terms off the main diagonal and these terms can be negative for the profiles of $\H_4$ and $\Wav_4$.  Despite the higher sensitivity of these two strategies, the overall error for queries that involve many counts, such as $\w=[1,1,1,1]$, approaches that of $\I_4$. The small dimension of these examples hides the extremely poor performance of $\I_n$ on queries that involve many counts. 
\end{example}

The next example uses Prop. \ref{prop:error:workload} to gain insight into the behavior of some natural strategies for answering a workload of queries.

\begin{example}If $\A$ is the identity matrix, then Prop. \ref{prop:error:workload} implies that the total error will depend only on the workload, since the sensitivity of $\I_n$ is 1: 
$$\totalerror{\I_n}{\W}=(\frac{2}{\epsilon^2}) \; \tr (\W^t\W).$$
Here the trace of $\W^t\W$ is the sum of squared coefficients of each query.  This will tend to be a good strategy for workloads that sum relatively few counts. Assuming the workload is full rank, we can use the workload itself as a strategy, i.e. $\A=\W$. Then Prop. \ref{prop:error:workload} implies that the total error is: 
$$\totalerror{\W}{\W}=(\frac{2}{\epsilon^2}) \; \sens{\W}^2 \;n. $$
since $\tr( \inv{(\W^t\W)} \W^t\W)=\tr(\I_n)=n$.  In this case the trace term is low, but the strategy will perform badly if $\sens{\W}$ is high.  Note that if $\W$ is $m \times n$, the total error of the Laplace mechanism for $\W$ will be $((\frac{2}{\epsilon^2}) \sens{\W}^2 m)$, which is worse than the matrix mechanism whenever $m>n$. 

In some sense, good strategies fall between the two extremes above: they should have sensitivity less than the workload but a trace term better than the identity strategy.

\end{example}


\subsection{Error profile decomposition}

Because an error profile matrix $\M$ is equal to $\ep{\A}$ for some $\A$, it has a number of special properties.  $\M$ is always a square ($n \times n$) matrix, it is symmetric, and even further, it is always a {\em positive definite} matrix.  Positive definite matrices $\M$ are such that $\w\M\w^t>0$ for all non-zero vectors $\w$.  In our setting this means that the profile term is always positive, as expected.  Furthermore, the function $f=\w \M \w^t$ is a quadratic function if $\w$ is viewed as a vector of variables. Then the function $f$ is an elliptic paraboloid over $n$-dimensional space.  If we consider the equation $\w\M\w^t=1$, this defines an ellipsoid centered at the origin (the solutions to this equation are the queries whose profile term is one).  We can think of the error function of strategy $\A$ as a scaled version of this paraboloid, where the scale factor is $(\frac{2}{\epsilon^2}) \; \sens{\A}^2$.  

To gain a better understanding of the error profile, we consider its decomposition.  Recall that a matrix is orthogonal if its transpose is its inverse.

\begin{definition}[Decomposition of Profile]
Let $\M$ be any $n \times n$ positive definite matrix.  The spectral decomposition of $\M$ is a factorization of the form  $\M= \P_\M \D_\M \P_\M^t$, where $\D_\M$ is an $n \times n$ diagonal matrix containing the eigenvalues of $\M$, and $\P_\M$ is an orthogonal $n \times n$ matrix containing the eigenvectors of $\M$. 
\end{definition}

Thus the matrices $\D_\M$ and $\P_\M$ fully describe the error profile.  They also have an informative geometric interpretation.  The entries of the diagonal matrix $\D_\M$ describe the relative stretch of the axes of the ellipsoid.  The matrix $\P_\M$ is orthogonal, representing a rotation of the ellipsoid.  In high dimensions, a set of common eigenvalues mean that the ellipsoid is spherical with respect to the corresponding eigen-space.  For example, the profile of $\I_4$ is fully spherical (all eigenvalues are one), but by choosing unequal eigenvalues and a favorable rotation, the error is reduced for certain queries.  


\subsection{Strategy matrix decomposition}

Despite the above insights into tuning the error profile, the matrix mechanism requires the choice of a strategy matrix, not simply a profile.  Next we focus on the relationship between strategies and their profile matrices.  

We will soon see that more than one strategy can result in a given error profile.  Accordingly, we define the following equivalence on query strategies:

\begin{definition}[Profile Equivalence] Two query matrices $\A$ and $\B$ are {\em profile equivalent} if their error profiles match, i.e. $\ep{\A}=\ep{\B}$.
\end{definition}

A key point is that {\em two profile equivalent strategies may have different sensitivity}.  If $\A$ and $\B$ are profile equivalent, but $\A$ has lower sensitivity, then strategy $\A$ dominates strategy $\B$: the estimate for {\em any} query will have lower error using strategy $\A$.

\begin{example} Recall that strategies $\Wav_4$ and $\H_4$ both have sensitivity 3.  This is not the minimal sensitivity for strategies achieving either of these error profiles.  A square matrix $\H'$, profile equivalent to $\H_4$, is shown in Figure \ref{fig:improved}(a).  This matrix has sensitivity $\sens{\H'}=2.896$.  A matrix $\Wav'$, profile equivalent to  $\Wav_4$ is shown in Figure \ref{fig:improved}(b).  This matrix has sensitivity $\sens{\Wav'}=2.210$.  
\end{example}

\begin{figure}[t]
\centering
\subfigure[$\H'$ is profile equivalent to $\H_4$, but $\sens{\H'}=2.896$ while $\sens{\H_4}=3$ ]{
$ \H' =
\begin{bmatrix} 
 -1.32 & -1.32 & -1.32 & -1.32 \\
  0.87 &  0.87 & -0.87 & -0.87 \\
 -0.71 &  0.71 &  0.00 &  0.00 \\
  0.00 &  0.00 & -0.71 & 0.71 \\
\end{bmatrix}
$ 
}
\hfill
\subfigure[$\Wav'$ is profile equivalent to $\Wav_4$, but $\sens{\Wav'}=2.210$ while $\sens{\Wav_4}=3$ ]{
$\Wav'=
\begin{bmatrix} 
1.73 & 0.58 & 0.00 & 0.00 \\
0.00 & 1.63 & 0.00 & 0.00 \\
0.00 & 0.00 & 1.73 & 0.58 \\
0.00 & 0.00 & 0.00 & 1.63 \\
\end{bmatrix}	
$

}
\caption{\label{fig:improved} When two strategies $\A$ and $\B$ are profile equivalent, the one with lower sensitivity dominates.}
\end{figure}

To analyze strategy matrices we again use matrix decomposition, however because a strategy matrix may not be symmetric, we use the singular value decomposition.

\begin{definition}[Decomposition of Strategy]
Let $\A$ be any $m \times n$ query strategy.  The singular value decomposition (SVD) of $\A$ is a factorization of the form $\A = \Q_\A \D_\A \P_\A^t$ such that $\Q_\A$ is a $m \times m$ orthogonal matrix, $\D_\A$ is a $m \times n$ diagonal matrix and $\P_\A$ is a $n \times n$ orthogonal matrix.  When $m > n$, the diagonal matrix $\D_\A$ consists of an $n \times n$ diagonal submatrix combined with $\vect{0}^{(m-n) \times n}$.
\end{definition}

The following theorem shows how the decompositions of the error profile and strategy are related.  It explains exactly how a strategy matrix determines the parameters for the error profile, and it fully defines the set of all profile equivalent strategies.

\begin{theorem} \label{thm:profile-strategy}
Let $m \geq n$ and let $\M$ be any $n \times n$ positive definite matrix decomposed as $\M= \P_\M \D_\M \P_\M^t$ where $\D_\M=diag(\lambda_1 \dots \lambda_n)$.  Then for any $m \times n$ matrix $\A$, the following are equivalent:
\begin{enumerate}
\item[(i)] $\A$ achieves the profile $\M$, that is $\ep{\A}=\M$; 
\item[(ii)] There is a decompostion of $\A$ into $\A=\Q_\A \D_\A \P^t_\A$ where $\Q_\A$ is an $m \times m$ orthogonal matrix, $\D_\A$ is an $m \times n$ matrix equal to $diag(1/\sqrt{\lambda_1} \dots 1/\sqrt{\lambda_n})$ plus $\vect{0}^{(m-n)\times n}$, and $\P_\A=\P_\M$. 
\end{enumerate}
\end{theorem}

\begin{proof}
Given (i), let $\D'=diag(\sqrt{\lambda_1} \dots \sqrt{\lambda_n})$, and since $\P_\M$ is an orthogonal matrix $\P_\M^t=\inv{\P_\M}$. Then 
\begin{align*}
\D'^t\P_\M^t\A^t\A\P_\M\D'&= \D'^t\P_\M^t\inv{\M}\P_\M\D' \\ &=\D'^t\P_\M^t\P_\M\inv{\D_\M}\P_\M^t\P_\M\D'=\I_n.
\end{align*}
Thus $\A=\Q'_\A\inv{\D'}\P_\M^t$ where the $\Q'_\A$ is an $m\times n$ matrix whose column vectors are unit length and orthogonal to each other. Let $\Q_\A$ be an $m\times m$ orthogonal matrix whose first $n$ columns are $\Q'_\A$. Let $\D_\A$ be an $m\times n$ matrix equal to $\inv{\D'}$ plus $\vect{0}^{(m-n)\times n}$, which is equivalent to $diag(1/\sqrt{\lambda_1} \dots 1/\sqrt{\lambda_n})$ plus $\vect{0}^{(m-n)\times n}$. Then
\[\Q_\A\D_\A\P_\M^t=\Q'_A\inv{\D'}\P_\M^t=\A.\]

Given (ii) we have $\A=\Q_\A \D_\A \P_\M^t$ and we first compute $\A^t\A=(\Q_\A \D_\A \P_\M^t)^t(\Q_\A \D_\A \P_\M^t)$ $=(\P_\M \D_\A^t \Q_\A^t )$ $(\Q_\A \D_\A \P_\M^t)$ $=(\P_\M \D_\A^t \D_\A \P_\M^t)$.  Note that while $\D_\A$ may be $m \times n$, $\D_\A^t \D_\A$ is an $n \times n$ diagonal matrix equal to $diag(1/\lambda_1 \dots 1/\lambda_n)$.  Then $\inv{(\A^t\A)}=\inv{(\P_\M \D_\A^t \D_\A \P_\M^t)}=$ $(\P_\M \inv{(\D_\A^t \D_\A)} \P_\M^t)$.  Then since $\inv{(\D_\A^t \D_\A)}=diag(\lambda_1 \dots \lambda_n)=\D_\M$ we conclude that $\inv{(\A^t\A)}=\P_\M \D_\M \P_\M^t$.
\end{proof}

This theorem has a number of implications that inform our optimization problem.  First, it shows that given any error profile $\M$, we can construct a strategy that achieves the profile.  We do so by decomposing $\M$ and constructing a strategy $\A$ from its eigenvectors (which are contained in $\P_\M$ and inherited by $\P_\A$) and the diagonal matrix consisting of the inverse square root of its eigenvalues (this is $\D_\A$, with no zeroes added).  We can simply choose $\Q$ as the $n \times n$ identity matrix, and then matrix $\D_\A \P^t_\A$ is an $n \times n$ strategy achieving $\M$.  

Second, the theorem shows that there are many such strategies achieving $\M$, and that {\em all} of them can be constructed in a similar way.  There is a wrinkle here only because some of these strategies may have more than $n$ rows.  That case is covered by the definition of $\D_\A$, which allows one or more rows of zeroes to be added to the diagonal matrix derived from the eigenvalues of $\M$.  Adding zeroes, $\D_\A$ becomes $m \times n$, we choose any $m \times m$ orthogonal matrix $\Q_\A$, and we have an $m \times n$ strategy achieving $\M$. 

Third, the theorem reveals that the key parameters of the error profile corresponding to a strategy $\A$ are determined by the $\D_\A$ and $\P_\A$ matrices of the strategy's decomposition.  For a fixed profile, the $\Q_\A$ of the strategy has no impact on the profile, but does alter the sensitivity of the strategy.  Ultimately this means that choosing an optimal strategy matrix requires determining a profile ($\D_\A$ and $\P_\A$), and choosing a rotation ($\Q_\A$) that controls sensitivity.  The rotation should be the one that minimizes sensitivity, otherwise the strategy will be dominated. 

We cannot find an optimal strategy by optimizing either of these factors independently. Optimizing only for the sensitivity of the strategy severely limits the error profiles possible (in fact, the identity matrix is a full rank strategy with least sensitivity).  If we optimize only the profile, we may choose a profile with a favorable ``shape'' but this could result in a prohibitively high least sensitivity.  Therefore we must optimize jointly for the both the profile and the sensitivity and we address this challenge next.

%% file: optimizing.tex

\section{Optimization} \label{sec:optimizing}

In this section, we provide techniques for determining or approximating optimal query strategies for the matrix mechanism, and we also give some heuristic strategies that may improve existing strategies.  We first state our main problem.

\begin{problem}[minError] \label{problem:minError}
Given a workload matrix $\W$, find the strategy $\A$ that minimizes $\totalerror{\A}{\W}$.
\end{problem}

The $\minerror$ problem is difficult for two reasons. First, 
 the sensitivity $\sens{\A}$ is the maximum function applied to the $L_1$ norms of column vectors, which is not differentiable. Second, we do not believe $\minerror$ is expressible as a convex optimization problem since the set of all query strategies that support a given workload $\W$ is not convex: if $\A$ supports $\W$, then $-\A$ also supports $\W$ but $\frac{1}{2}(\A+(-\A))=\mathbf{0}$ does not support $\W$. 
 
In Section~\ref{sec:solving:optimization} we 
show that $\minerror$ can be expressed as a semidefinite program with rank constraints.  While rank constraints make the semidefinite program non-convex, there are algorithms that can solve such problems by iteratively solving a pair of related semidefinite programs.

Though the set of all query strategies $\A$ that support a given workload $\W$ is not convex, the set of all possible matrices $\A^t\A$ is convex. In Sec.~\ref{sec:sub:approx} we provide two approaches for finding approximate solutions based on bounding  $\sens{\A}$ by a function of $\A^t\A$ rather than a function of $\A$. While each technique results in a strategy, they essentially select a profile and a default rotation $\Q$.  Error bounds are derived by reasoning about the default rotation.  It follows that both of these approximations can be improved considering rotations $\Q$ that reduce sensitivity.  Therefore we also consider a secondary optimization problem.

\begin{problem}[minSensitivity] \label{problem:minSens}
Given a query matrix $\A$, find the query matrix $\B$ that is profile equivalent to $\A$ and has minimum sensitivity.
\end{problem}

Unfortunately this subproblem is still not a convex problem, since the set of all query matrices that are profile equivalent is also not convex. Notice $\A$ is profile equivalent to $-\A$ but is not profile equivalent to $\frac{1}{2}(\A+(-\A))=\mathbf{0}$. Again the problem can be expressed as an SDP with rank constraints as it is shown in Sec.~\ref{sec:minsens}.

\subsection{Solution to the $\minerror$ Problem} \label{sec:solving:optimization}
In this section we formulate the $\minerror$ problem for an $n\times n$ workload $\W$. It is sufficient to focus on $n\times n$ workloads because Proposition \ref{prop:error:workload} shows that the strategy that minimizes the total error for some workload $\W$ also minimizes the total error for any workload $\V$ such that $\V^t\V=\W^t\W$.  Therefore, if given an $m \times n$ workload for $m > n$, we can use spectral decomposition to transform it into an equivalent $n\times n$ workload.

\mh{Chao: fix inconsistent vector notation}

\mh{Chao: don't use $\mathbf{Y}$, that's the wavelet}

\begin{theorem}\label{thm:solving:l1}
Given an $n\times n$ workload $\W$, Program~\ref{prog:minerr} is a semidefinite program with rank constraint whose solution is the tuple $(\A, \mathbf{C}, \vect{u}, \mathbf{Z})$ and the $m\times n$ strategy $\A$ minimizes $\totalerror{\A}{\W}$ among all $m\times n$ strategies.
\begin{algorithm}[t]
\caption{Minimizing the Total Error}\label{prog:minerr}
{\small
\begin{align}
\mbox{Given:          } & \W\in\mathbb{R}^{n\times n}\nonumber\\
\mbox{Minimize:      } &u_1+u_2+\ldots + u_n\nonumber\\
\mbox{Subject to:    } 
& \mbox{For } i\in [n]: \mbox{$\vect{e}_i$ is the $n$ dimensional column vector whose}\nonumber\\
& \mbox{$i^{th}$ entry is 1 and remaining entries are 0.}\nonumber\\
& ~~~~ \left[\begin{array}{ccc}2\mathbf{I}_m & -\A\inv{\W} & \mathbf{0}\\
-(\A\inv{\W})^t & \inv{(\W^t)}\mathbf{Z}\inv{\W} & \vect{e}_i \\ \mathbf{0} & \vect{e}_i^t & u_i\end{array}\right]\succeq 0\label{eqn:sdpconstraint}\\
& \mbox{For } i\in [n], j\in [m]: \nonumber\\
& ~~~~ c_{ji}\geq a_{ji},~~~~ c_{ji}\geq -a_{ji},~~~~\sum_{k=1}^{m} c_{ki} \leq 1 
\label{eqn:sensconstraint}\\
& \rank\left(\left[\begin{array}{cc}\mathbf{I}_m & \A\\ \A^t & \mathbf{Z}\end{array}\right]\right)=m\label{eqn:rankconstraint}
\end{align}}
\end{algorithm}
\end{theorem}

\begin{proof}
In Program~\ref{prog:minerr}, $u_1+\ldots+u_n$ is an upper bound on the total error (modulo constant factors).  The rank constraint in Eq.~(\ref{eqn:rankconstraint}) makes sure that $\mathbf{Z}=\A^t\A$.  

The semidefinite constraint, Eq.~(\ref{eqn:sdpconstraint}), ensures that $u_i$ is an upper bound on twice the error of the $i^{th}$ query in the workload, ignoring for the moment the sensitivity term.
$$
u_i \geq 2(\W\inv{(\A^t\A)}\W^t)_{ii}
$$
To show this, let $\mathbf{X}$ be the $(m+n)\times(m+n)$ upper left submatrix of the matrix in Eq.~(\ref{eqn:sdpconstraint}), substituting $\A^t\A$ for $\mathbf{Z}$:
\[\mathbf{X}=\left[\begin{array}{cc}2\mathbf{I}_m & -\A\inv{\W}\\ -(\A\inv{\W})^t & \inv{(\W^t)}\A^t\A\inv{\W} \end{array}\right],\]
and then 
\[\inv{\mathbf{X}}=\left[\begin{array}{cc}\W\inv{(2\mathbf{I}_m-\A\A^\plus)}\W^t & (\W\A^\plus)^t\\ \W\A^\plus & 2\W\inv{(\A^t\A)}\W^t \end{array}\right].\]

The semidefinite constraints in Eq.~(\ref{eqn:sdpconstraint}) are equivalent to:
\[\forall i,\; u_i\geq(\inv{\mathbf{X}})_{m+i,m+i}= 2(\W\inv{(\A^t\A)}\W^t)_{ii}.\]

Thus, minimizing $u_1+\ldots+u_n$ is equivalent to minimizing the trace of $\W\inv{(\A^t\A)}\W^t$.  To make $u_1+\ldots+u_n$ a bound on the total error, we must show that $\sens{\A} = 1$.  The constraints in Eq.~(\ref{eqn:sensconstraint}) ensure that $\sens{\A} \leq 1$.  To see that $\sens{\A} \geq 1$, observe that
$\inv{(k\mathbf{X})}=\frac{1}{k}\inv{\mathbf{X}}$.  So $u_1 + \dots u_n$ is minimized when $\sens{\A} = 1$ because otherwise
we can multiply $\mathbf{X}$ (which contains $\A$) by a constant to make $u_1+\ldots+u_n$ smaller. Above all, we have
\begin{align*}
u_1+u_2+\ldots+u_n&=2\sum_{i=1}^n (\W\inv{(\A^t\A)}\W^t)_{ii}\\
&=2\tr (\W\inv{(\A^t\A)}\W^t)\\
&=2\sens{\A}^2 \; \tr (\W\inv{(\A^t\A)}\W^t)\\
&=\epsilon^2 \totalerror{\A}{\W}.
\end{align*}
with $\epsilon$ fixed.
\end{proof}

%

Thus Theorem~\ref{thm:solving:l1} provides the best strategy to the $\minerror$ problem with at most $m$ queries.  Observe that if the optimal strategy has $m' < m$ queries, then Program~\ref{prog:minerr} will return an $m \times n$ matrix with $m - m'$ rows of $0$s.  In addition, if the workload contains queries with coefficients in $\{ -1, 0, 1 \}$, we can show that $n^2$ is upper bound on the number of queries in the optimal strategy.

Dattorro \cite{dattorro2005convex} shows that solving a semidefinite program with rank constraints can be converted into solving two semi\-definite programs iteratively. The convergence follows the widely used trace heuristic for rank minimization. We are not aware of results that quantify the number of iterations that are required for convergence. However, notice it takes $O(n^3)$ time to solve a semidefinite program with an $n\times n$ semidefinite constraint matrix and in Program~\ref{prog:minerr}, there are $n$ semidefinite constraint matrices with size $m+n$, which can be represented as a semidefinite constraint matrix with size $n(m+n)$. Thus, the complexity of solving our semidefinite program with rank constraints is at least $O(m^3n^3)$.

\subsection{Approximations to the $\minerror$ problem} \label{sec:sub:approx}
\eat{Since it is difficult to solve $minError$ problem accurately, some heuristics can be applied to get approximate solutions. One of the straightforward ideas is to convert the $m\times n$ workload $\W$ into a $n\times n$ $\A$ strategy while keep the same error profile. Theorem~\ref{thm:profile-strategy} indicates that such strategy can be found by the singular value decomposition of $\W$: let $\W=\Q_\W\D_\W\P^t_\W$ and $\D'_\W$ be the first $n$ rows of $\D_\W$. Then $\D'_\W\P^t_\W$ is a strategy as we want.

Though singular value decomposition provides an easy way to generate an alternative query strategy, the optimality of the strategy can not be guaranteed. }

As mentioned above, the $\minerror$ problem can be simplified by bounding the sensitivity of $\A$ with some properties of $\A^t\A$. Here we introduce two approximation methods that use this idea and can be computed efficiently: the $L_2$ approximation (\ref{sec:solving:l2approx}), and the singular value bound approximation (\ref{sec:solving:eigenapprox}). Error bounds on both methods can be measured by providing upper bounds to $\sens{\A}$.

\subsubsection{$L_2$ approximation}
\label{sec:solving:l2approx}
Note that the diagonal entries of $\A^t\A$ are the squared $L_2$ norms of column vectors of $\A$. For sensitivity, recall that we are interested in the maximum $L_1$ norm of the column vectors of $\A$. 
This observation leads to the following approaches:  we can either use the $L_2$ norm as an upper bound to the $L_1$ norm, or we can relax the definition of differential privacy by measuring the sensitivity in terms of $L_2$ rather than $L_1$.

%

~ \\ \noindent {\bf Using $L_2$ norm as an upper bound to $L_1$ norm.}
Instead of $\minerror$, we can solve the following $L_2$ approximation problem. We use $\Ltwo{\A}$ to denote the maximum $L_2$ norm of column vectors of $\A$.

\begin{problem}[$L_2$ approximation] \label{problem:L2approx}
Given a workload matrix $\W$, find the strategy $\A$ that minimizes \[\Ltwo{\A}^2\tr (\W \inv{(\A^t\A)} \W^t).\]
\end{problem}

According to the basic property of $L$ norms, for any vector $\vect{v}$ of dimension $n$, $\Ltwo{\vect{v}}\leq \Lone{\vect{v}} \leq \sqrt{n}\Ltwo{\vect{v}}$. Therefore we can bound the approximation rate of the $L_2$ approximation.

\begin{theorem}\label{thm:L2approx}
Given a workload $\W$, let $\A$ be the optimal solution to the $minError$ problem and $\A'$ be the optimal solution to the $L_2$ approximation. Then \[\totalerror{\A'}{\W}\leq n\totalerror{\A}{\W}.\]
\end{theorem}
Notice the $L_2$ bound is equal to the $L_1$ bound if all queries in strategy $\A$ are uncorrelated, so that the $L_2$ approximation gives the optimal strategy if the optimal strategy only contains uncorrelated queries such as $\I_n$.

~ \\ \noindent {\bf Relaxing the definition of differential privacy.}
$L_2$ norms can also be applied by relaxing the definition of $\epsilon$-differential privacy into $(\epsilon,\delta)$-differential privacy, which is defined as following:

\begin{definition}[$(\epsilon,\delta)$-differential privacy] A randomized algorithm $\alg$ is $(\epsilon,\delta)$-differentially private if for any instance $I$, any $I' \in \nbrs(I)$, and any subset of outputs $S \subseteq Range(\alg)$, the following holds:
\[
Pr[ \alg(I) \in S] \leq \exp(\epsilon) \times Pr[ \alg(I') \in S] +\delta
\]		
where the probability is taken over the randomness of the $\alg$.
	\end{definition}

$(\epsilon, \delta)$-differential privacy can be achieved by answering each query in strategy $\A$ with i.i.d Gaussian noise:
\begin{restatable}{theorem}{Relaxdp}\label{thm:l2diffpriv}
Let $\W$ be a qu\-ery matrix consisting of $m$ queries, and let $\b_\delta$ be a length-$m$ column vector consisting of independent samples from a Gaussian distribution with scale $N(0, 8\ln(2/\delta))$. Then for $\epsilon\leq 8\ln(2/\delta)$, $\delta\leq 1$, the randomized algorithm $\LM$ that outputs the following vector is $(\epsilon, \delta)$-differentially private:
$$\LM_\delta(\W,\x) = \W\x + (\frac{\Ltwo{\W}}{\epsilon})\b_\delta$$
\end{restatable}


Recall the proof of Proposition~\ref{prop:error:single} and apply it to Theorem~\ref{thm:l2diffpriv}.  Minimizing the total error under $(\epsilon, \delta)$-differential privacy is equivalent to solving  Problem~\ref{problem:L2approx}. 

A semidefinite program (Program~\ref{prog:L2approx}) can be used to solve Problem~\ref{problem:L2approx}. For a given solution $\mathbf{X}$ of Program~\ref{prog:L2approx}, any $n\times n$ matrix $\A$ such that $\mathbf{X}=\A^t\A$ is a valid solution to Problem~\ref{problem:L2approx}. Moreover, when $\delta$ is given, the Gaussian noise added in the $(\epsilon, \delta)$-differential privacy is $\Theta({\epsilon^2}\Ltwo{\A}^2)$. According to the relationship between $L_1$ and $L_2$ norm, the Laplace noise added in the $\epsilon$-differential privacy is  at least $\Omega({\epsilon^2}\Ltwo{\A}^2)$, which indicates relaxing the definition of differential privacy will always lead to better utility.
\begin{algorithm}[t]
\caption{$L_2$ approximation}\label{prog:L2approx}
\begin{align*}
\mbox{Given:          } & \W\in\mathbb{R}^{n\times n}.\\
\mbox{Minimize:      } &u_1+u_2+\ldots+u_n.\\
\mbox{Subject to:    } &\mbox{For } i\in [n]: \mbox{$\vect{e}_i$ is the $n$ dimensional column vector whose}\nonumber\\
& \mbox{$i^{th}$ entry s 1 and other entries are 0.}\nonumber\\
& \left[\begin{array}{cc}\mathbf{X} & \vect{e}_i \\ \vect{e}_i^t & u_i\end{array}\right]\succeq 0;\\
& (\W^t\mathbf{X}\W)_{ii}\leq 1,\quad i\in [n].
\end{align*}
\end{algorithm}

\subsubsection{Singular value bound approximation}
\label{sec:solving:eigenapprox}

Another way to bound the $L_1$ sensitivity is based on its geometric properties. Remember the matrix $\A$ can be reperesented by its singular value decomposition $\A=\Q_{\A}\D_{\A}\P_{\A}^t$. Let us consider the geometry explanation of the sensitivity. The sensitivity of $\A$ can be considered as the radius of minimum $L_1$ ball that can cover all column vectors of $\A$, and column vectors of $\A$ lay on the ellipsoid
\[\phi_\A: \x^t\Q_\A^t\inv{(\D_\A^t\D_\A)}\Q_\A\x=1.\]
Let $\sens{\phi_\A}$ denotes radius of the minimum $L_1$ ball that covers the ellipsoid $\phi_\A$. Notice all the column vectors of $\A$ are contained in $\phi_\A$, which indicates $\sens{\A}\leq\sens{\phi_\A}$. The minimum sensitivity that can be achieved by the strategies that are profile equivalent to $\A$ can be bounded as following:
\[\min_{\B\,:\,\B^t\B=\A^t\A}\sens{\B}\leq\min_{\B\,:\,\B^t\B=\A^t\A}\sens{\phi_\B}.\]
The matrix $\B$ that is profile equivalent to $\A$ and has the minimum $\sens{\phi_\B}$ is given by the theorem below.

\begin{restatable}{theorem}{EigenBound}\label{thm:eigenbound}
Let $\A$ be a matrix with singular value decomposition $\A=\Q_{\A}\D_{\A}\P_{\A}^t$ and $\delta_1, \delta_2, \ldots, \delta_n$ be its singular values. Then
\begin{align}
\underset{\B\,:\,\B^t\B=\A^t\A}{\operatorname{argmin}}\sens{\phi_\B}&=\D_\A\P_\A^t, \nonumber\\
\min_{\B\,:\,\B^t\B=\A^t\A}\sens{\phi_\B}&= \sqrt{\delta_1^2+\delta_2^2+\ldots+\delta_n^2} \leq \sqrt{n}\sens{\A}. \label{eqn:eigenbound}
\end{align}
\end{restatable}

Using the singular value bound in Theorem~\ref{thm:eigenbound} to substitute for the $L_1$ sensitivity, the $minError$ problem can be converted to the following approximation problem.
\begin{problem}[Singular value bound approximation] \label{problem:eigenapprox}
Given a workload matrix $\W$, find the strategy $\A$ that minimizes \[(\delta_1^2+\delta_2^2+\ldots+\delta_n^2)\tr(\W\inv{(\A^t\A)}\W^t),\] 
where $\delta_1, \delta_2, \ldots, \delta_n$ are singular values of $\A$.
\end{problem}

The singular value bound approximation has a closed-form solution.
\begin{restatable}{theorem}{EigenApprox}
	\label{thm:eigen_approx}
	Let $\W$ be the workload matrix with singular value decomposition $\W=\Q_\W \D_\W \P^t_\W$ and $\delta'_1, \delta'_2,\ldots,\delta'_n$ be its singular values. The optimal solution $\D_{\A}$, $\P_{\A}$ to the singular value bound approximation is to let $\P_{\A} = \P_\W$ and $\D_\A=diag(\sqrt{\delta'_1},\sqrt{\delta'_2},\ldots,\sqrt{\delta'_n})$.
\end{restatable}
The solution in Theorem~\ref{thm:eigen_approx} is very similar to the strategy metioned at the end of Sec.~\ref{sec:analysis} that matches $\P_\A$ to $\P_\W$ and $\D_\A$ be $diag(\delta'_1,\delta'_2,\ldots,\delta'_n)$. We use a slightly different $\D_\A$ so as to provide an guaranteed error bound based on Theorem~\ref{thm:eigenbound}.

\begin{theorem}
Given a workload $\W$, let $\A$ be the optimal solution to the $minError$ problem and $\A'$ be the optimal solution to the singular value bound approximation. Then \[\totalerror{\A'}{\W}\leq n\totalerror{\A}{\W}.\]
\end{theorem}

\eat{\subsubsection{Singluar value decomposition method}
\label{sec:solving:svd}
Given a workload $\W$, remember that 
\[\tr(\W\inv{(\A^t\A)}\W^t)=\tr(\inv{(\A^t\A)}\W^t\W),\]
which gives us the intuition that we can choose the $\A$ such that $\inv{(\A^t\A)}\W^t\W=\mathbf{I}$. Moreover, if the number of queries in $\W$ is much larger than the domain size $n$, it is better to ask a smaller set of queries instead. One of the straightforward ways to achieve the goals above is to perform a singluar value decomposition on $\W=\Q_\W\D_\W\P^t_\W$ and let $\D'_\W$ be the first $n$ rows of $\D_\W$. Then $\D'_\W\P^t_\W$ is a strategy that fits the two properties above.
}

\subsection{Augmentation Heuristic}\label{sec:aug}

We formalize below the following intuition: as far as the error profile is concerned, additional noisy query answers can never detract from query accuracy as they must have some information content useful to one or more queries.  Therefore the error profile can never be worse after augmenting the query strategy by adding rows.  

\begin{restatable}{theorem}{augment}{\sc(Augmenting a strategy)} \label{thm:augment}
Let $\A$ be a query strategy with full rank and consider a new strategy $\A'$ obtained from $\A$ by adding the additional rows of strategy $\B$, so that $\A'=\left[
\begin{smallmatrix}
\A \\ \B
\end{smallmatrix}
\right]$. For any query $\w$, we have:
\[
\w^t(\A'^t\A')^{-1}\w \leq \w^t(\A^t\A)^{-1}\w 
\]
Further, $\w^t(\A'^t\A')^{-1}\w = \w^t(\A^t\A)^{-1}\w$ only for the queries in the set $\{ \A^t\A\w \; | \; \B\w=0\}$, which is non-empty if and only if $\B$ does not have full column rank.
\end{restatable}

The proof is included in \reffull{Appendix \ref{app:deficient}}.

This improvement in the error profile may have a cost---namely, augmenting $\A$ with strategy $\B$ may lead to a strategy $\A'$ with greater sensitivity than $\A$.  A heuristic that follows from Theorem~\ref{thm:augment} is to augment strategy $\A$ only by completing deficient columns, that is, by adding rows with non-zero entries only in columns whose absolute column sums are less the sensitivity of $\A$.  In this case the augmentation does not increase sensitivity and is guaranteed to strictly improve accuracy for any query with a non-zero coefficient in an augmented column. 

Our techniques could also be used to reason formally about augmentations that do incur a sensitivity cost.  We leave this as future work, as it is relevant primarily to an interactive differentially private mechanism which is not our focus here.

\subsection{Minimizing the sensitivity}
\label{sec:minsens}

We now return to Problem \ref{problem:minSens} which finds the strategy with least sensitivity that results in a given profile.  This problem is important whenever one has a specific profile in mind (e.g. the profile of strategy $\H_n$ or $\Wav_n$), or when one used another method to compute a desired profile (e.g. the approximation method from Section \ref{sec:sub:approx}).  Recall that for a fixed error profile, the profile-equivalent strategies are determined by the choice of a rotation matrix $\Q$ which then determines the sensitivity of the strategy.  The following theorem formulates the problem of minimizing the sensitivity into a semidefinite program with rank constraint.

\begin{theorem}\label{thm:minsens}
Given an error profile $\M$, Program~\ref{prog:minsens} is a semidefinite program with rank constraint that outputs a square matrix $\A$ such that $\inv{(\A^t\A)}=\M$ and such that the sensitivity of $\A$ is minimized.

\begin{algorithm}[t]
\caption{Minimizing the sensitivity}\label{prog:minsens}
\begin{align*}
\mbox{Given:          } & \M\in\mathbb{R}^{n\times n}.\\
\mbox{Minimize:      } &r.\\
\mbox{Subject to:    } & \mbox{For } i\in [n], j\in [m]: \\
& ~~~~ c_{ji}\geq a_{ji},~~~~ c_{ji}\geq -a_{ji},~~~~\sum_{k=1}^{m} c_{ki} \leq r\\ 
& \rank\left(\left[\begin{array}{cc}\mathbf{I}_n & \A\\ \A^t & \inv{\M}\end{array}\right]\right)=n.
\end{align*}
\end{algorithm}
\end{theorem}

%% file: applications.tex

\section{Applications} \label{sec:apps}

In this section we use our techniques to analyze and improve existing approaches. We begin by analyzing two techniques proposed recently~\cite{xiao2010differential,Hay:2010Boosting-the-Accuracy}.  Both strategies can be seen as instances of the matrix mechanism, each using different query strategies designed to support a workload consisting of all range queries.  Although both techniques can support multidimensional range queries, we focus our analysis on one dimensional range queries, i.e. interval queries with respect to a total order over $dom(\bb)$.  

We will show that the seemingly distinct approaches have remarkably similar behavior: they have low (but not minimal) sensitivity, and they are highly accurate for range queries but much worse for queries that are not ranges. We describe these techniques briefly and how they can each be represented in matrix form.

In the {\em hierarchical} scheme proposed in \cite{Hay:2010Boosting-the-Accuracy}, the query strategy can be envisioned as a recursive partitioning of the domain.  We consider the simple case of a binary partitioning, although higher branching factors were considered in \cite{Hay:2010Boosting-the-Accuracy}.  First we ask for the total sum over the whole domain, and then ask for the count of each half of the domain, and so on, terminating with counts of individual elements of the domain.  For a domain of size $n$ (assumed for simplicity to be a power of 2), this results in a query strategy consisting of $2n-1$ rows.  We represent this strategy as matrix $\H_n$, and $\H_4$ in Fig.~\ref{fig:query-matrix} is a small instance of it. 

In the {\em wavelet} scheme, proposed in \cite{xiao2010differential}, query strategies are based on the Haar wavelet. For one dimensional range queries, the technique can also be envisioned as a hierarchical scheme, asking the total query, then asking for the difference between the left half and right half of the domain, continuing to recurse, asking for the difference in counts between each binary partition of the domain at each step.\footnote{We note that the technique in \cite{xiao2010differential} is presented somewhat differently, but that the differences are superficial.  The authors use queries that compute averages rather than sums, and their differentially private mechanism adds scaled noise at each level in the hierarchy.  We prove the equivalence of that construction with our formulation $\Wav_n$ in \reffull{App.~\ref{app:haar}}.}  This results in $n$ queries---fewer than the hierarchical scheme of~\cite{Hay:2010Boosting-the-Accuracy}.  The matrix corresponding to this strategy is the matrix of the Haar wavelet transform, denoted $\Wav_n$, and $\Wav_4$ in Fig.~\ref{fig:query-matrix} is a small instance of it.  

Thus $\H_n$ is a rectangular $(2n-1) \times n$ strategy, with answers derived using the linear regression technique, and $\Wav_n$ is an $n \times n$ strategy with answers derived by inverting the strategy matrix.  As suggested by the examples in earlier sections, these seemingly different techniques have similar behavior.  We analyze them in detail below, proving new bounds on the error for each technique, and proving new results about their relationship to one another.  We also include $\I_n$ in the analysis, which is the strategy represented by the dimension $n$ identity matrix, which asks for each individual count.

\vspace{2ex}
\subsection{Geometry of $\I_n$, $\H_n$ and $\Wav_n$}

Recall from Section \ref{sec:analysis} that the decomposition of the error profile of a strategy explains its error.  The decomposition of $\I_n$ results in a $\D$ that is itself the identity matrix.  This means the error profile is spherical. To understand the shape and rotation of the error profiles for $\Wav_n$ and $\H_n$ we provide a complete analysis of the decomposition, but leave the details in \reffull{the Appendix~\ref{app:eigen}.  The eigenvalues and eigenvectors are shown in Table~\ref{tab:eigens} of Appendix \ref{app:eigen}}. Their eigenvalue distributions are remarkably similar. Each has $\log n+1$ distinct eigenvalues of geometrically increasing frequency.  The actual eigenvalues of $\H_n$ are smaller than those of $\Wav_n$ by exactly one throughout the increasing sequence, except the largest eigenvalue: it is equal to the second largest eigenvalue in $\Wav_n$, but it has a distinct value in $\H_n$.  Finally, the smallest eigenvalue of either approach is 1 and the ratio between their corresponding eigenvalues is in the range $[\frac{1}{2}, 2]$.

For sensitivity, it is clear that $\sens{\I_n}=1$ for all $n$.  Intuitively, this sensitivity should be minimal since the columns of $I_n$ are axis aligned and orthogonal, and any rotation of $\I_n$ can only increase the $L_1$ ball containing the columns of $\I_n$.  This intuition can be formalized by considering the relationship between the $L_1$ norm and the $L_2$ norm stated in Section \ref{sec:solving:l2approx}.  No strategy profile equivalent to $I_n$ can have lower sensitivity, since $\sens{I_n}=\Ltwo{I_n}=1$.

On the other hand, the sensitivity of $\Wav_n$ and $\H_n$ is not minimal, suggesting that there exist strategies that dominate both of them.  We have $\sens{\Wav_n}=\sens{\H_n}=\log_2n + 1$.  In addition we find that their $L_2$ norms are also equal: $\Ltwo{\Wav_n} = \Ltwo{\H_n} = \sqrt{\log_2 n + 1}$.  This $L_2$ norm is a lower bound on the sensitivity of profile equivalent strategies for both $\H_n$ and $\Wav_n$.  We do not know if there are profile equivalent strategies that achieve this sensitivity lower bound for these strategies.  We can, however, improve on the sensitivity of both.  As an example, Fig.~\ref{fig:improved} shows profiles equivalent to $\H_4$ and $\Wav_4$ with improved sensitivity.  Through our decomposition of $\H_n$ and $\Wav_n$ we have derived modest improvements on the sensitivity in the case of arbitrary $n\geq 8$: $\log n + 0.64$ for $\H_n$, which is the sensitivity of its decomposition, and $\log n + 2\sqrt{2} -4$ for $\Wav_n$, which is achieved by applying some minor modifications to its decomposition. We suspect it is possible to find rotations of $\H_n$ and $\Wav_n$ that improve more substantially on the sensitivity.


\subsection{Error analysis for $\I_n$,$\H_n$ and $\Wav_n$}

In this section we analyze the total and worst case error for specific workloads of interest. We focus on two typical workloads: $\Wrang$, the set of all range queries, and $\Wbool$, which includes arbitrary predicate queries, since it consists of all linear queries 0-1 queries.  Note that attempting to use either of these workloads as strategies leads to poor results: the sensitivity of $\Wrang$ is $O(n^2)$ while the sensitivity of $\Wbool$ is $O(2^n)$.


In the original papers describing $\H_n$ and $\Wav_n$ \cite{Hay:2010Boosting-the-Accuracy, xiao2010differential}, both techniques are shown to have worst case error bounded by $O(\log^3 n)$ on $\Wrang$.  Both papers resort to experimental analysis to understand the distribution of error across the class of range queries.  We note that our results allow error for any query to be analyzed analytically.  

It follows from the similarity of eigenvectors and eigenvalues of $\H_n$ and $\Wav_n$ that the error profiles are asymptotically equivalent to one another.  We thus prove a close equivalence between the error of the two techniques:

\begin{restatable}{theorem}{Compareerrorwh}\label{thm:compareerrorwh} 
For any linear counting query $\w$, 
\[\frac{1}{2}\error{\Wav}{\w}\leq \error{\H}{\w}\leq 2\error{\Wav}{\w}.\]
\end{restatable}

Note that this equivalence holds for the hierarchical strategy with a branching factor of two.  Higher branching factor can lower the error rates of the hierarchical strategy compared with the wavelet technique.

Next we summarize the maximum and total error for these strategies.  The following results tighten known bounds for $\Wrang$, and show new bounds for $\Wbool$. The proof of the following theorem can be found in \reffull{Appendix \ref{app:error}}.

\begin{restatable}[Maximum and Total Error]{theorem}{MaximumError} The maximum and total error on workloads $\Wrang$ and $\Wbool$ using strategies $\H_n, \Wav_n,$ and $\I_n$ is given by:
\[\begin{array}{c|ccc}
\mbox{\sc MaxError} & \H_n & \Wav_n & \I_n \\ \hline
\Wrang & \Theta(\log^3 n / \epsilon^2) & \Theta(\log^3 n / \epsilon^2) & \Theta(n / \epsilon^2)\\
\Wbool & \Theta(n\log^2 n / \epsilon^2) & \Theta(n\log^2 n/ \epsilon^2) & \Theta(n/ \epsilon^2)\\
\end{array}
\]
\[\begin{array}{c|ccc}
\mbox{\sc TotalError} & \H_n & \Wav_n & \I_n \\
\hline \Wrang & \Theta(n^2\log^3 n / \epsilon^2) & \Theta(n^2\log^3 n/ \epsilon^2) & \Theta(n^3/ \epsilon^2)\\
\Wbool & \Theta(n2^n\log^2n/ \epsilon^2) & \Theta(n2^n\log^2n/ \epsilon^2) & \Theta(n2^n/ \epsilon^2)
\end{array}
\]
\end{restatable}
While $\H_n$ and $\Wav_n$ achieve similar asymptotic bounds, their error profiles are slightly different (as suggested by previous examples for $n=4$).  As a result, $\H_n$ tends to have lower error for larger range queries, while $\Wav_n$ has lower error for unit counts and smaller range queries.

%% file: related-work.tex

\section{Related Work}
\label{sec:related}

Since differential privacy was first introduced~\cite{dwork2006calibrating}, it has been the subject of considerable research, as outlined in recent surveys~\cite{dwork2008differential,dwork2009differential,Dwork:2010A-firm-foundation}.  

Closest to our work are the two techniques, developed independently, for answering range queries over histograms.  Xiao et al.~\cite{xiao2010differential} propose an approach based on the Haar wavelet; Hay et al.~\cite{Hay:2010Boosting-the-Accuracy} propose an approach based on hierarchical sums and least squares.  The present work unifies these two apparently disparate approaches under a significantly more general framework (Section~\ref{sec:query_answering}) and uses the framework to compare the approaches (Section~\ref{sec:apps}).  While both approaches are instances of the matrix mechanism, the specific algorithms given in these papers are more efficient than a generic implementation of the matrix mechanism employing matrix inversion.  Xiao et al. also extend their wavelet approach to nominal attributes and multi-dimensional histograms.

Barak et al.~\cite{barak2007privacy} consider a Fourier transformation of the data to estimate low-order marginals over a set of attributes.  
The main utility goal of~\cite{barak2007privacy} is integral consistency: the numbers in the marginals must be non-negative integers and their sums should be consistent across marginals.  Their main result shows that it is possible to achieve integral consistency (via Fourier transforms and linear programming) without significant loss in accuracy.  We would like to use the framework of the matrix mechanism to further investigate optimal strategies for workloads consisting of low-order marginals.



%



Blum et al.~\cite{blum2008a-learning} propose a mechanism for accurately answering queries for an arbitrary workload (aka query class), where the accuracy depends on the VC-dimension of the query class.  However, the mechanism is inefficient, requiring exponential runtime.  They also propose an efficient strategy for the class of range queries, but this approach is less accurate than the wavelet or hierarchical approaches discussed here (see Hay et al.~\cite{Hay:2010Boosting-the-Accuracy} for comparison).


Some very recent works consider improvements on the Laplace mechanism for multiple queries.  Hardt and Talwar~\cite{Hardt:2010On-the-Geometry-of-Differential} consider a very similar task based on sets of linear queries. They propose the $k$-norm mechanism, which adds noise tailored to the set of linear queries by examining the shape to which the linear queries map the $L_1$ ball.  They also show an interesting lower bound on the noise needed for satisfying differential privacy that matches their upper bound up to polylogarithmic factors assuming the truth of a central conjecture in convex geometry.  But the proposed $k$-norm mechanism can be inefficient in practice because of its requirement of sampling uniformly from high-dimensional convex bodies.  Furthermore, the techniques restrict the number of queries to be less than $n$~(the domain size)\eat{and require the coefficients of the linear queries to lie in $[-1,1]$ GM:I don't want to mention that}.  A notable difference in our approach is that our computational cost is incurred for finding the query strategy.  Once a strategy is found, our mechanism is as efficient as the Laplace mechanism.  For stable or recurring workloads, optimization needs only to be performed once.

Roth and Roughgarden~\cite{Roth:2010The-Median-Mechanism:} consider the interactive setting, in which queries arrive over time and must be answered immediately without knowledge of future queries. They propose the median mechanism which improves upon the Laplace mechanism by deriving answers to some queries from the noisy answers already received from the private server. The straightforward implementation of the median mechanism is inefficient and requires sampling from a set of super-polynomial size, while a more efficient polynomial implementation requires weakening the privacy and utility guarantees to average-case notions (i.e., guarantees hold for most but not all input datasets).

The goal of optimal experimental design \cite{Pukelsheim:1993Optimal-Design} is to produce the best estimate of an unknown vector from the results of a set of experiments returning noisy observations.  Given the noisy observations, the estimate is typically the least squares solution.  The goal is to minimize error by choosing a subset of experiments and a frequency for each.  A relaxed version of the experimental design problem can be formulated as a semi-definite program \cite{boyd2004convex}.  While this problem setting is similar to ours, a difference is that the number and choice of experiments is constrained to a fixed set.  In addition, although experimental design problems can include costs associated with individual experiments, modeling the impact of the sensitivity of experiments does not fit most problem formulations. Lastly, the objective function of most experimental design problems targets the accuracy of individual variables (the $\x$ counts), rather than a specified workload computed from those counts.

%
%
%
%

%% file: conclusion.tex

\section{Conclusion}

We have described the matrix mechanism, which derives answers to a workload of counting queries from the noisy answers to a different set of strategy queries.  By designing the strategy queries for the workload, correlated sets of counting queries can be answered more accurately.  We show that the optimal strategy can be computed by iteratively solving a pair of semidefinite programs, and we use our framework to understand two recent techniques targeting range queries.

While we have formulated the choice of strategy matrix as an optimization problem, we have not yet generated optimal---or approximately optimal---solutions for specific workloads.  Computing such optimal strategies for common workloads would have immediate practical impact as it could boost the accuracy that is efficiently achievable under differential privacy.  We also plan to apply our approach to interactive query answering settings, and we would like to understand the conditions under which optimal strategies in our framework can match known lower bounds for differential privacy. 

\subsection*{Acknowledgements}
We would like to thank the anonymous reviewers for their helpful comments. The NSF supported authors Hay, Li, Miklau through IIS-0643681 and CNS-0627642, author Rastogi through IIS-0627585, and author McGregor through CCF-0953754.

%% file: appendix.tex


\section{Scaling query strategies}
The matrix mechanism always adds identically-distributed noise to each query in the strategy matrix. In addition, in Section~\ref{sec:optimizing}, the sensitivity of all considered strategies is bounded by 1. In this section, we demonstrate that those constraints do not limit the power of the mechanism or the generality of the optimization solutions.

\subsection{Scalar multiplication of query strategy}

Multiplying a query strategy by a scalar value (which scales every coefficient of every query in the strategy) does not change the error of any query.  The sensitivity term is scaled, but it is compensated by a scaling of the error profile.

\begin{proposition}[Error equivalence under scalar multiplication] 
Given query strategy $\A$, and any real scalar $k \neq 0$, the error for any query $\w$ is equivalent for strategies $\A$ and $k\A$. That is, $\forall \w, \error{\A}{\w\estx} = \error{k\A}{\w\estx}$.
\end{proposition}
\begin{proof}
It is easy to see that $\sens{k\A}=k\sens{\A}$. Thus, for any query $\w$, we have 
\begin{align*}
\error{k\A}{\w\estx}&=\sens{k\A}^2\w \inv{((k\A)^t(k\A))}\w^t \\ &=k^2\sens{\A}^2\w\inv{(k^2\A^t\A)}\w^t\\ 
&=k^2\sens{\A}^2\w\frac{1}{k^2}\inv{(\A^t\A)}\w^t\\
&=\Delta^2_\A\w(\A^t\A)^{-1}\w^t\\
&=\error{\A}{\w\estx}.
\end{align*}
\end{proof}

\subsection{Strategies with unequal noise -- scaling rows} 
\label{app:unequal_noise}

\def\E{\vect{E}}
\def\R{\vect{R}}

The strategies described in this paper add equal noise to each query; i.e., the independent Laplace random variables have equal scale.  It is sufficient to focus on equal-noise strategies because, as the following proposition shows, any unequal-noise strategy can be simulated by an equal-noise strategy.

\begin{proposition}[Simulating unequal noise strategies]
	Let $\alg_\A$ be an unequal-noise strategy that returns $\y = \A \x + \E$ where $\E$ is a $m$-length vector of independent Laplace random variables with unequal scale.  Let $b_i$ denote the scale of the $i^{th}$ Laplace random variable in the vector. There exists an equal-noise strategy such that its output, $\y'$, has the same distribution as $\y$.
\end{proposition}

\begin{proof}
	Let $\R$ be an $m \times m$ diagonal matrix with $r_{ii} = b / b_i$ for an arbitrary $b$.  Let the equal noise strategy be defined as $\alg_\B = \B \x + \vect{Lap(b)}$ where the query matrix $\B = \R \A$.  Let $\y' = \R^{-1}( \B \x + \vect{Lap(b)} ) $. 
	
The claim is that $\y$ and $\y'$ have the same distribution.  Vector $\y'$ can be expressed as:
\begin{align*}
	\y' &= \R^{-1}( \B \x + \vect{Lap(b)} ) \\
	&=  \R^{-1}\B \x + \R^{-1} \vect{Lap(b)} \\
	&= \A \x + \R^{-1} \vect{Lap(b)}
\end{align*}
Observe that $\R^{-1} \vect{Lap(b)}$ is an $m \times 1$ vector where the $i^{th}$ entry is equal to $\frac{b_i}{b} Lap(b)$.  The quantity $\frac{b_i}{b} Lap(b)$ follows a Laplace distribution with scale $b_i$, and is therefore equal in distribution to the $i^{th}$ entry of $\E$, which is $Lap(b_i)$.  Therefore $\y'$ is equal in distribution to $\y$.
\end{proof}

\section{Relaxation of Differential Privacy}
$(\epsilon,\delta)$-differential privacy is introduced in Section~\ref{sec:solving:l2approx}. Here we formally prove the amount of gaussian noise required to achieve $(\epsilon, \delta)$-differential privacy.

\Relaxdp*

\begin{proof}
Let $I$ and $I'$ be neighboring databases, and then their corresponding vectors $\x$ and $\x'$ differ in exact one component. Notice
\[\max_{\Lone{\x-\x'}=1}\Ltwo{\W\x-\W\x'}=\Ltwo{\W},\]
consider adding column vector $(\Ltwo{\W}/\epsilon)\b$ to $\W\x$. Let \[\sigma=(\Ltwo{\W}/\epsilon)2\sqrt{2\ln(2/\delta)}.\] 
For any vector equals to $\W\x+\vect{z}$, we have
\begin{align*}
& \frac{\Pr[\W\x+ (\frac{\Ltwo{\W}}{\epsilon})\b =\W\x+ \vect{z}]}{\Pr[\W\x'+ (\frac{\Ltwo{\W}}{\epsilon})\b = \W\x+\vect{z}]} \\
=&\frac{e^{-\frac{1}{2\sigma^2}(\vect{z}^t\vect{z})}}{ e^{-\frac{1}{2\sigma^2}((\W\x-\W\x'+\vect{z})^t(\W\x-\W\x'+\vect{z}))}} \\
=&e^{\frac{1}{2\sigma^2}(\Ltwo{\W\x-\W\x'}^2-2\vect{z}^t(\W\x-\W\x'))}\\
\leq& e^{\frac{\epsilon^2}{16\ln(2/\delta)}+\frac{\vect{z}^t(\W\x-\W\x')}{\sigma^2}}\\
\leq& e^{\frac{\epsilon}{2}+\frac{\vect{z}^t(\W\x-\W\x')}{\sigma^2}}.
\end{align*}
Let $\mathcal{Z}=\{\W\x+\vect{z}|\frac{\vect{z}^t(\W\x-\W\x')}{\sigma^2}>\frac{\epsilon}{2}\}$, to guarantee $(\epsilon, \delta)$-differential privacy, we only need to proof
\begin{align*}
\delta & \geq \Pr[\W\x+ (\frac{\Ltwo{\W}}{\epsilon})\b \in \mathcal{Z}]\\
&=\Pr[\b^t(\W\x-\W\x')\geq 4\Ltwo{\W}\ln(2/\delta)].
\end{align*}
Notice the entries of random vector $\b$ are independent varibles following $N(0, 8\ln(2/\delta))$ and $\b^t(\W\x-\W\x')$ can be considered as a weighted sum of all the entries of $\b$, $\b^t(\W\x-\W\x')$ follows $N(0, 8\Ltwo{\W\x-\W\x'}^2\ln(2/\delta))$. Let $z$ be a variable that follows $N(0,1)$ Then
\begin{align*}
&\Pr[\b^t(\W\x-\W\x')\geq 4\Ltwo{\W}\ln(2/\delta)]\\
=&\Pr[2\Ltwo{\W\x-\W\x'}\sqrt{2\ln(2/\delta)}z\geq 4\Ltwo{\W}\ln(2/\delta)]\\
=&\Pr[z\geq \frac{\Ltwo{\W}}{\Ltwo{\W\x-\W\x'}}\sqrt{2\ln(2/\delta)}]\\
\leq & \Pr[z\geq \sqrt{2\ln(2/\delta)}].
\end{align*}
Notice that
\[\Pr[z\geq x]<\frac{1}{x}\frac{1}{\sqrt{2\pi}}e^{-\frac{x^2}{2}}
,\]
we have
\[\Pr[z\geq \sqrt{2\ln(2/\delta)}]<\frac{\delta}{2\sqrt{\pi\ln(2/\delta)}}<\delta.\]
Thus  the randomized algorithm $\LM$ follows $(\epsilon, \delta)$-differentially privacy.
\end{proof}

\section{Singular Value Bound Approximation} \label{app:opt}

In this section we theoretically compute the approximation rate of the singular value bound approximation and the optimized solution under the singular value bound approximation.
%
%

\begin{lemma}\label{lem:tangentplane} Given an ellipsoid defined by $\x^t\mathbf{Z}\x = 1$ and a vector $\v$, $\v^t\x=\sqrt{\v^t\inv{\mathbf{Z}}\v}$ is
a tangent hyperplane of the ellipsoid.
\end{lemma}

\begin{proof}
For any point $\y$ on the ellipsoid, the tangent hyperplane of the ellipsoid on $\y$ is $\y^t\A\x=1$. Consider
a tangent hyperplane of the ellipsoid: $\v^t\x = k$ where $k$ is an unknown constant, there exists a point $\x_0$
on the ellipsoid such that $\x_0^t\A=\frac{\v^t}{k}$. Therefore $\x_0=\frac{\inv{\mathbf{Z}}\v}{k}$. Since $\x_0^t \mathbf{Z} \x_0 = 1$, we know
\[1 = \x_0^t \mathbf{Z} \x_0 
= (\frac{1}{k}\v^t \inv{\mathbf{Z}}) \mathbf{Z} (\frac{1}{k}\inv{\mathbf{Z}}\v ) 
= \frac{1}{k^2}\v^t\inv{\mathbf{Z}}\v.
\]
Therefore $k=\sqrt{\v^t\inv{\mathbf{Z}}\v}$.
\end{proof}

\EigenBound*

\begin{proof}
For any strategy $\B$, the ellipsoid $\phi_{\B}$ must tangent with diamond with radius $\sens{\phi_\B}$. With out lose of generality, let us assume it is tangent to the hyperplane $(1,1,\ldots,1)\x =\sens{\phi_\B}$ and $(a_1, \ldots, a_n)\x \leq  \sens{\phi_\B}$, here $a_i = {1, -1}$. Let $\B=\Q_\B\D_\A\P_\A^t$ be the singular value decomposition of $\B$ and let $\Psi=\Q_\B^t\inv{(\D_\A^t\D_\A)}\Q_\B$ to simplify the notation. According to Lemma~\ref{lem:tangentplane}, 
\begin{align*}
(1,1,\ldots,1)\Psi(1, 1\ldots,1)^t \geq (a_1, \ldots, a_n)\Psi(a_1, \ldots, a_n)^t.
\end{align*}
In particular,
\begin{align*}
(1,1,..1)\Psi(1, 1\ldots,1)^t \geq (-1, 1, 1, 1, \ldots, 1)\Psi (-1, 1, 1, 1, \ldots, 1)^t,
\end{align*}
which means $\psi_{12} + \psi_{13} + \ldots + \psi_{1n} = \sum_{i=1}^n \psi_{1i} - \psi_{11} \geq 0$. 
Similarly, we can show for any j we have $\sum_{i=1}^n \psi_{ji} - \psi_{jj} \geq 0$.
Therefore
\begin{align*}
&(1,1,..1)\Psi(1, 1\ldots,1)^t\\
&= \sum_i \sum_j \psi_{ij}\\
&= \sum_j \psi_{jj} + \sum_j (\sum_i \psi_{ji} - \psi_{jj} )\\
&\geq \sum_j \psi_{jj}\\
&= \tr(\Psi)\\
&= \tr(\Q_\B^t\inv{(\D_\A^t\D_\A)}\Q_\B)\\
&= \tr( \Q_\B\Q_\B'^t\inv{(\D_\A^t\D_\A)} )\\
& =\tr( \inv{(\D_\A^t\D_\A)} )
\end{align*}

Since $\D_\A$ is fixed, the minimize can be achieved in case that $\Phi$ is a diagonal matrix, which indicates that $\Q_\B=\I$. Therefore,
\begin{align*}
\B &= \D_\A\P_\A^t\\
\sens{\phi_\B} &= \sqrt{(1,1,..1)\inv{\Psi}(1, 1\ldots,1)^t}=\sqrt{(1,1,..1)\D_\A^t\D_\A(1, 1\ldots,1)^t}=\sqrt{\delta_1^2+\delta_2^2+\ldots+\delta_n^2},
\end{align*}
where $\delta_1, \delta_2, \ldots, \delta_n$ are the singular values of $\A$.

Moreover, notice the fact that the sum of square of $L_2$ norm of all the columns of $\B$ is same as the sum of square of $L_2$ norm of all the rows of $\B=\D_\A\P_\A^t$. Since $\P^\A$ is a rotation matrix it does not change the $L_2$ norm of rows in $\D_\A$, which is $\delta_1^2+\delta_2^2+\ldots+\delta_n^2$. Notice $\B$ has $n $ columns in total, there exists a column of $\B$ whose $L_2$ norm is at least $\sqrt{\frac{\delta_1^2+\delta_2^2+\ldots+\delta_n^2}{n}}$ so that $\sens{\phi_\B}\leq\sqrt{n}\Ltwo{\B}$. Since $\Ltwo{\B}=\Ltwo{\A}$, $\Ltwo{\A}\leq\sens{\A}$, we know $\sens{\phi_\B}\leq\sqrt{n}\sens{\A}$.
\end{proof}

\EigenApprox*

\begin{proof}
Recall the total error with the singular value approximation:
\[(\delta_1^2+\delta_2^2+\ldots+\delta_n^2)\tr(\W\inv{(\A^t\A)}\W^t),\]
where $\delta_1, \delta_2, \ldots, \delta_n$ are the singular values of $\A$. Notice 
\begin{align*}
\tr(\W\inv{(\A^t\A)}\W^t) &=\tr(\inv{(\A^t\A)}\W^t\W)\\
&= \tr(\P_\A^t\inv{(\D_\A^t\D_\A)}\P_\A\P^t_\W\D_\W^t\D_\W\P_\W)\\
&= \tr(\D_\W^t(\P_\A\P_\W^t)^t\inv{(\D_\A^t\D_\A)}(\P_\A\P_\W^t)\D_\W^t\D_\W),
\end{align*}
and $\P_\A$ does not influence the singular value approximation to the sensitivity. $\P_\A$ can be arbitrary orthogonal matrix and then $\P_\A\P_\W^t$ can be arbitrary orthogonal matrix as well. Then the best $\P_\A\P_\W^t$ is set to be the one that minimizes the error on estimating $\D_\W$ with given $\D_\A$. Since $\D_\W$ is actually a set of queries over individual buckets, the best strategy to estimate it is also queries over individual buckets, which means $\P_\A\P_\W^t=\I$. Since $\P_\W$ is an orthogonal matrix, $\P_\A=\P_\W$.  Then, the total error with singular value approximation is:
\begin{align*}
(\delta_1^2+\delta_2^2+\ldots+\delta_n^2)\tr(\inv{(\D_\A^t\D_\A)}\D_\W^t\D_\W)
&=(\delta_1^2+\delta_2^2+\ldots+\delta_n^2)(\frac{\delta_1'^2}{\delta_1^2}+\frac{\delta_2'^2}{\delta_2^2}+\ldots+\frac{\delta_n'^2}{\delta_n^2})\\
&=(\delta_1'+\delta_2'+\ldots+\delta_n')^2
\end{align*}
To achieve the lower bound given by the equality above, it requires that for each $i$, the ratio between $\frac{\delta_i'}{\delta_i}$ and $\delta_i$ to be constant, which means $\delta_i=c\sqrt{\delta'_i}$ for some constant $c$. Since a constant multiple does not change the strategy, we let $c=1$ and then have the theorem proved. 
\end{proof}

\section{Completing deficient columns} \label{app:deficient}

Here we complete the proof of the theorem about augmenting the deficient columns in Section~\ref{sec:aug}.
\begin{proposition}\label{prop:eigenproperty}
For any square matrix $\A$, if $v$ is an eigenvector of $\A$ with eigenvalue $\lambda$, $v$ is an eigenvector of $k\I+\A$ with eigenvalue $k+\lambda$. If $\A$ is invertible, $v$ is an eigenvector of $\A^{-1}$ with eigenvector $\frac{1}{\lambda}$.
\end{proposition}
\begin{proof}
Since $(k\I+\A)v=k\I \v+\A v=k\v+\lambda \v=(k+\lambda)\v$, we know $v$ is an eigenvector of $k\I+\A$ with eigenvalue $k+\lambda$. If $\A$ is invertible, we know $\lambda\neq 0$ otherwise there exists a non-zero vector $v$ such that $\A \v=0$, which contradicts with the fact that $\A$ is invertible. Moreover, since $\A \v=\lambda \v$, $\A^{-1}\v=\frac{1}{\lambda}\v$, $\v$ is an eigenvector of $\A^{-1}$ with eigenvector $\frac{1}{\lambda}$.
\end{proof}

\augment*
\begin{proof}
Since $\A$ is a query plan with full column rank, there exists an invertible square matrix $\Q$ such that $\A^t\A=\Q^t\Q$. Moreover, notice that $\A'^t\A=\A^t\A+\B^t\B$, the theorem we are going to prove is equivalent to the following statement: the matrix $(\Q^t\Q)^{-1} - (\Q^t\Q+\B^t\B)^{-1}$ is positive semi-definite. Since $\Q$ is an invertible matrix, for any query $\w$, 
\begin{align*}
&\w^t((\Q^t\Q)^{-1} - (\Q^t\Q+\B^t\B)^{-1})\w\\	
=&((\Q^t)^{-1}\w)^t(\I-\Q(\Q^t\Q+\B^t\B)^{-1}\Q^t)((\Q^t)^{-1}\w).
\end{align*}
Therefore it is enough to show $\I-\Q(\Q^t\Q+\B^t\B)^{-1}\Q^t$ is positive semi-definite, which is equivalent to prove that all eigenvalues of $\Q(\Q^t\Q+\B^t\B)^{-1}\Q^t$ are less than or equal to $1$ according to Proposition~\ref{prop:eigenproperty}. Furthermore, since $\Q(\Q^t\Q+\B^t\B)^{-1}\Q^t$ is invertible, according to Proposition~\ref{prop:eigenproperty}, the statement that all eigenvalues of $\Q(\Q^t\Q+\B^t\B)^{-1}\Q^t$ are less than or equal to $1$ is equivalent to the  statement that all the eigenvalue of its inverse matrix, $(\Q^t)^{-1}(\Q^t\Q+\B^t\B)\Q^{-1}$, is larger than or equal to $1$. Notice $(\Q^t)^{-1}(\Q^t\Q+\B^t\B)\Q^{-1}=\I+(\B\Q^{-1})^t(\B\Q^{-1})$, we can apply Proposition~\ref{prop:eigenproperty} again and to prove the eigenvalues of $(\B\Q^{-1})^t(\B\Q^{-1})$ are non-negative. Since for any vector $\v$, $\v^t(\B\Q^{-1})^t(\B\Q^{-1})\v\geq 0$, we know is semi-positive definite, hence all its eigenvalues are non-negative.

Moreover, according to Proposition~\ref{prop:eigenproperty}, $(\Q^t)^{-1}\w$ is an eigenvector of $\I-\Q(\Q^t\Q+\B^t\B)^{-1}\Q^t$ with eigenvalue $0$ if and only if it is an eigenvector of $\Q(\Q^t\Q+\B^t\B)^{-1}\Q^t=(\I+ (\B\Q^{-1})^t(\B\Q^{-1}))^{-t}$ with eigenvalue $1$, which is equivalent with the fact that $(\Q^t)^{-1}\w$ is an eigenvector of $(\B\Q^{-1})^t(\B\Q^{-1})$ with eigenvalue $0$. Furthermore, notice the fact that $\A^t\A$ is an invertible matrix and $(\B\Q^{-1})^t(\B\Q^{-1})(\Q^t)^{-1}\w=0$ is equivalent to $(\B\Q^{-1})(\Q^t)^{-1}\w=\B(\Q^t\Q)^{-1}\w=\B(\A^t\A)^{-1}\w=0$. Therefore $(\B\Q^{-1})^t(\B\Q^{-1})$ has eigenvalue $0$ if and only if $\B$ does not have full column rank. When $\B$ does not have full column rank, the set $\w_{\B}=\{\w|\B\w=0\}$ is not empty, and the set of all non-zero queries $\w$ such that $\w^t(\A'^t\A')^{-1}\w=\w^t(\A^t\A)^{-1}\w$ can be represented as $\{\A^t\A\w | \w\in \w_{\B}\}=\{\A^t\A\w | \w\in \w\B=0\}$.
\end{proof}


\section{Representing the Haar wavelet technique}
\label{app:haar}

\def\estHaar{\estx_{Haar}}
\def\estW{\estx_{\Wav_n}}

The representation of Haar wavelet queries in Section~\ref{sec:apps} is different from their original presentation in Xiao et al.~\cite{xiao2010differential}. The following theorem shows the equivalence of both representations.

\begin{proposition}[Equivalence of Haar wavelet representations]
	Let $\estHaar$ denote the estimate derived from the Haar wavelet approach of Xiao et al.~\cite{xiao2010differential}.  Let $\estW$ denote the estimate from asking query $\W_n$.  Then $\estHaar$ and $\estW$ are equal in distribution, i.e., $Pr[ \estHaar \leq x] = Pr[ \estW \leq x]$ for any vector $x$.
\end{proposition}

\begin{proof}
	Given vector $\x$, the Haar wavelet is defined in terms of a binary tree over $\x$ such that the leaves of the tree are $\x$.  
	
Each node in the tree is associated with a coefficient.
Coefficient $c_i$ is defined as $c_i = (a_L - a_R)/2$ where $a_L$ ($a_R$) is the average of the leaves in the left (right) subtree of $c_i$.  Each $c_i$ is associated with a weight $\mathcal{W}(c_i)$ which is equal to the number of leaves in subtree rooted at $c_i$.  (In addition, there is a coefficient $c_0$ that is the equal to the average of $\x$ and $\mathcal{W}(c_0) = n$).

An equivalent definition for $c_i$ is $c_i = \sum_{j=1}^{n} x_j z_i(j)$ where for $i > 0$,
	\begin{eqnarray*}
	z_i(j) & = &
	\left\{
	\begin{array}{l}
	1/\mathcal{W}(c_i),~~\mbox{if $j$ is in the left subtree of $c_i$}\\
	-1/\mathcal{W}(c_i),~~\mbox{if $j$ is in the right subtree of $c_i$}\\
	0,~~\mbox{otherwise}\\
	\end{array}
	\right.
	\end{eqnarray*}
For $i = 0$, then $z_i(j)$ is equal to $1/\mathcal{W}(c_0)$ for all $j$.

Let $\A$ be a matrix where $a_{ij} = z_i(j)$.  The $i^{th}$ row of $\A$ corresponds to coefficient $c_i$.  Since there are $n$ coefficients, $\A$ is an $n \times n$ matrix.

The approach of \cite{xiao2010differential} computes the following $\y_{Harr} = \A \x + \E$ where $\E$ is an $n \times 1$ vector such that each $\E_i$ is an independent sample from a Laplace distribution with scale $b_i = \frac{1 + \log n}{\epsilon \mathcal{W}(c_i)}$.  Observe that $\E$ can be equivalently represented as:
\[
\E = \inv{\R} \left(\frac{1 + \log n}{\epsilon}\right) \b
\]
where $\R$ is an $n \times n$ diagonal matrix with $r_{ii} = \mathcal{W}(c_i)$.
The estimate for $\x$ is then equal to:
\begin{align*}
	\estHaar &= \inv{\A} \y_{Harr} 
	= \x + \inv{\A} \E \\
	&= \x + \inv{\A} \inv{\R} \left(\frac{1 + \log n}{\epsilon}\right) \b \\
	&= \x + \inv{(\R\A)} \left(\frac{1 + \log n}{\epsilon}\right) \b
\end{align*}
We now describe an equivalent approach based on the matrix $\Wav_n$.  Observe that $\Wav_n = \R \A$.  The sensitivity of $\Wav_n$ is $\sens{\Wav_n} = 1 + \log n$.  Using the matrix mechanism, the estimate $\estW$ is:
\begin{align*}
	\estW &= \inv{\Wav_n} \left( \Wav_n \x + (\frac{\sens{\Wav_n}}{\epsilon})\b) \right) \\
	&= \x + \inv{\Wav_n} \frac{\sens{\Wav_n}}{\epsilon}\b \\
	&= \x + \inv{(\R\A)} \left(\frac{1 + \log n}{\epsilon}\right) \b
\end{align*}
\end{proof}

\section{Analysis of the hierarchical and wavelet strategies}\label{app:ana}
In Section~\ref{sec:apps} we demonstrated the results of applying the matrix mechanism to analyze hierarchical and wavelet scheme. In this section, we show the detailed analysis to those results.

\subsection{Eigen-decomposition of $\H_n$ and $\Wav_n$} \label{app:eigen}

The following eigen-decomposition shows the similarity between $\H_n$ and $\Wav_n$.
 
\begin{theorem}
Let $n$ be a power of $2$, so that $n=2^k$. The eigenvalues and their corrsponding eigienvectors of $\H_n^t\H_n$ and $\Wav_n^t\Wav_n$ are shown as Table~\ref{tab:eigens}. ($\mathbf{0}_{a\times b}$ and $\mathbf{1}_{a\times b}$ are the $a\times b$ matrices whose entries are all $0$ and $1$, respectively).

\begin{table*}[ht]
\[\begin{array}{c|c|c|c}
\mbox{eigenvalue of }\H_n & \mbox{eigenvalue of }\Wav_n & \mbox{order} & \mbox{eigenvector}\\
\hline 
\multirow{4}{*}{$1$} & \multirow{4}{*}{$2$} & \multirow{4}{*}{$2^{k-1}$} & [1, -1,\mathbf{0}_{1\times (2^k-2)}]\\
& & & [0, 0, 1, -1, \mathbf{0}_{1\times (2^k-4)}]\\
&  & & \cdots\\
& & & [\mathbf{0}_{1\times (2^k-2)}, 1, -1]\\
\hline
\multirow{4}{*}{$3$} & \multirow{4}{*}{$4$} & \multirow{4}{*}{$2^{m-2}$} & [\mathbf{1}_{1\times 2}, -\mathbf{1}_{1\times 2}, \mathbf{0}_{1\times (2^k-4)}]\\
& & & [\mathbf{0}_{1\times 4}, \mathbf{1}_{1\times 2}, -\mathbf{1}_{1\times 2}, \mathbf{0}_{1\times (2^k-8)}]\\
&  & & \cdots\\
& & & [\mathbf{0}_{1\times (2^k-4)}, \mathbf{1}_{1\times 2}, -\mathbf{1}_{1\times 2},]\\
\hline
\cdots & \cdots & \cdots & \cdots \\
\hline
\multirow{4}{*}{$2^k-1$} & \multirow{4}{*}{$2^k$} & \multirow{4}{*}{$2^{m-k}$} & [\mathbf{1}_{1\times 2^{k-1}}, -\mathbf{1}_{1\times 2^{k-1}}, \mathbf{0}_{1\times (2^k-2^k)}] \\
& & & [ \mathbf{0}_{1\times 2^k}, \mathbf{1}_{1\times 2^{k-1}}, -\mathbf{1}_{1\times 2^{k-1}}, \mathbf{0}_{1\times (2^k-2^{k+1})}]\\
& & & \cdots\\
& & & [ \mathbf{0}_{1\times (2^k-2^k)}, \mathbf{1}_{1\times 2^{k-1}}, -\mathbf{1}_{1\times 2^{k-1}}] \\
\hline
\cdots & \cdots & \cdots & \cdots \\
\hline
2^k -1& 2^k & 1 & [ \mathbf{1}_{1\times 2^{k-1}}, -\mathbf{1}_{1\times 2^{k-1}}]\\
\hline
2^{k+1}-1 & 2^k & 1 & \mathbf{1}_{1\times 2^{k}}\\
\hline
\end{array}
\]
\caption{\label{tab:eigens} Eigenvalues and eigenvectors for $\H_n^t\H_n$ and $\Wav_n^t\Wav_n$.}
\end{table*}
\end{theorem}

\begin{proof}
We will prove it by induction on $m$. When $k=1$,
\[\H_2^t\H_2=\left[\begin{array}{cc}2 & 1\\1 & 2\end{array}\right],\]
whose eigenvalues are $3$ and $1$ with eigenvector $[1, 1]$ and $[1, -1]$ respectively.
Suppose the conclusion is right for $\H_{2^{k-1}}$. Notice the fact that
\[\H_{2^k}=\left[\begin{array}{cc}
\mathbf{1}_{1\times 2^{k-1}} & \mathbf{1}_{1\times 2^{k-1}}\\
\H_{2^{k-1}} & 0 \\
0 & \H_{2^{k-1}}
\end{array}\right],\]
where $\mathbf{1}_{a\times b}$ is the $a\times b$ matrix whose entries are all $1$. Therefore we have
\begin{align*}
\H^t_{2^k}\H_{2^k}&=\left[\begin{array}{ccc}
\mathbf{1}_{2^{k-1}\times 1} & \H^t_{2^{k-1}} & 0 \\
\mathbf{1}_{2^{k-1}\times 1} & 0 & \H^t_{2^{k-1}}
\end{array}\right]\left[\begin{array}{cc}
\mathbf{1}_{1\times 2^{k-1}} & \mathbf{1}_{1\times 2^{k-1}}\\
\H_{2^{k-1}} & 0 \\
0 & \H_{2^{k-1}}
\end{array}\right] \\
&=\mathbf{1}_{2^{k}\times 1}\mathbf{1}_{1\times 2^{k}} + 
\left[\begin{array}{cc}
\H^t_{2^{k-1}} & 0 \\
0 & \H^t_{2^{k-1}}
\end{array}\right]\left[\begin{array}{cc}
\H_{2^{k-1}} & 0 \\
0 & \H_{2^{k-1}}
\end{array}\right]\\
&=\mathbf{1}_{2^{k}\times 2^k} + 
\left[\begin{array}{cc}
\H^t_{2^{k-1}}\H_{2^{k-1}} & 0 \\
0 & \H^t_{2^{k-1}}\H_{2^{k-1}}
\end{array}\right]
\end{align*}
Notice the case that $\mathbf{1}_{1\times 2^{k-1}}$ is a eigenvector of $\H^t_{2^{k-1}}\H_{2^{k-1}}$ with eigenvalue $2^k-1$, we have:
\begin{align*}
\H^t_{2^k}\H_{2^k}\mathbf{1}_{1\times 2^{k}}&=\mathbf{1}_{2^{k}\times 2^k}\mathbf{1}_{2^{k}\times 1} + 
\left[\begin{array}{cc}
\H^t_{2^{k-1}}\H_{2^{k-1}} & 0 \\
0 & \H^t_{2^{k-1}}\H_{2^{k-1}}
\end{array}\right]\mathbf{1}_{2^{k}\times 1}\\
&=\mathbf{1}_{2^{k}\times 2^k}\mathbf{1}_{2^{k}\times 1} + 
\left[\begin{array}{cc}
\H^t_{2^{k-1}}\H_{2^{k-1}}\mathbf{1}_{2^{k-1}\times 1} & 0 \\
0 & \H^t_{2^{k-1}}\H_{2^{k-1}}\mathbf{1}_{2^{k-1}\times 1}
\end{array}\right]\\
&=2^k\mathbf{1}_{2^{k}\times 1}+(2^k-1)\mathbf{1}_{2^{k}\times 1}=(2^{m+1}-1)\mathbf{1}_{2^{k}\times 1}
\end{align*}
\begin{align*}
\H^t_{2^k}\H_{2^k}\left[\begin{array}{c}\mathbf{1}_{2^{k-1}\times 1} \\ -\mathbf{1}_{2^{k-1}\times 1}\end{array}\right]&=\mathbf{1}_{2^{k}\times 2^k}\left[\begin{array}{c}\mathbf{1}_{2^{k-1}\times 1} \\ -\mathbf{1}_{2^{k-1}\times 1}\end{array}\right] + 
\left[\begin{array}{cc}
\H^t_{2^{k-1}}\H_{2^{k-1}} & 0 \\
0 & \H^t_{2^{k-1}}\H_{2^{k-1}}
\end{array}\right]\left[\begin{array}{c}\mathbf{1}_{2^{k-1}\times 1} \\ -\mathbf{1}_{2^{k-1}\times 1}\end{array}\right]\\
&=\left[\begin{array}{cc}
\H^t_{2^{k-1}}\H_{2^{k-1}}\mathbf{1}_{2^{k-1}\times 1} & 0 \\
0 & -\H^t_{2^{k-1}}\H_{2^{k-1}}\mathbf{1}_{2^{k-1}\times 1}
\end{array}\right]\\
&=(2^{k}-1)\left[\begin{array}{c}\mathbf{1}_{2^{k-1}\times 1} \\ -\mathbf{1}_{2^{k-1}\times 1}\end{array}\right]
\end{align*}
Moreover, for any eigenvector $\mathit{v}$ of $\H_{2^{k-1}}^t\H_{2^{k-1}}$ in Table~\ref{tab:eigens} other than $\mathbf{1}_{1\times 2^{k-1}}$, notice the fact that the sum of all entries in $\mathit{v}$ is $0$, denote the eigenvalue of $\mathit{v}$ as $\mu_{\mathit{v}}$, 
\begin{align*}
\H^t_{2^k}\H_{2^k}\left[\begin{array}{c}\mathit{v} \\ \mathbf{0}_{2^{k-1}\times 1}\end{array}\right]&=\mathbf{1}_{2^{k}\times 2^k}\left[\begin{array}{c}\mathit{v} \\ \mathbf{0}_{2^{k-1}\times 1}\end{array}\right] + 
\left[\begin{array}{cc}
\H^t_{2^{k-1}}\H_{2^{k-1}} & 0 \\
0 & \H^t_{2^{k-1}}\H_{2^{k-1}}
\end{array}\right]\left[\begin{array}{c}\mathit{v} \\ \mathbf{0}_{2^{k-1}\times 1}\end{array}\right]\\
&=\left[\begin{array}{cc}
\H^t_{2^{k-1}}\H_{2^{k-1}}\mathit{v} & 0 \\
0 & 0
\end{array}\right]=\mu_{\mathit{v}}\left[\begin{array}{c}\mathit{v} \\ \mathbf{0}_{2^{k-1}\times 1}\end{array}\right]
\end{align*}
Therefore $[\mathit{v}, 0]$ is a eigenvector of $\H_{2^k}^t\H_{2^k}$ with eigenvalue $\mu_{\mathit{v}}$, and we can show $[0, \mathit{v}]$ is a eigenvector of $\H_{2^k}^t\H_{2^k}$ with eigenvalue $\mu_{\mathit{v}}$ by a similar process. Above all, we proved that the vectors in Table~\ref{tab:eigens} are all eigenvectors of $\H_{2^k}^t\H_{2^k}$ with eigenvalues shown in the Table. Moreover, since any pair of eigenvectors from Table~\ref{tab:eigens} are orthogonal to each other and there are $2^k$ vectors in the Table, we know the table contains all eigenvalues of $\H_{2^k}^t\H_{2^k}$.

Similar as the case of $\H_{2^k}^t\H_{2^k}$, the eigenvectors and eigenvalues of $\Wav_{2^k}^t\Wav_{2^k}$ can be proved by induction on $n$ as well. When $k=1$,
\[\Wav_2^t\Wav_2=\left[\begin{array}{cc}2 & 0\\0 & 2\end{array}\right],\]
whose eigenvector $[1, 1]$ and $[1, -1]$ and both of them have eigenvalue $2$. Suppose Table~\ref{tab:eigens} gives eigenvalues and eigenvectors of $\Wav_{2^{k-1}}^t\Wav_{2^{k-1}}$. Notice the fact that
\begin{align*}
&\Wav_{2^k}=\left[\begin{array}{cc} \raisebox{-10pt}{$\Wav_{2^{k-1}}$} & \mathbf{1}_{1\times 2^{k-1}} \\
 & 0 \\
-\mathbf{1}_{1\times 2^{k-1}} & \raisebox{-10pt}{$\Wav_{2^{k-1}}$}\\
0 & 
\end{array}\right]\\
&\Wav^t_{2^k}\left[\begin{array}{c} \mathbf{1}_{1\times 2^k}\\ 0\end{array}\right]=\mathbf{1}_{2^k\times 2^k}
\end{align*}
Therefore,
\begin{align*}
\Wav^t_{2^k}\Wav_{2^k}&=\left[\begin{array}{cccc}\multicolumn{2}{c}{\Wav_{2^{k-1}}^t} & -\mathbf{1}_{2^{k-1}\times 1} & 0 \\ \mathbf{1}_{2^{k-1}\times 1} & 0 & \multicolumn{2}{c}{\Wav_{2^{k-1}}^t}\end{array}\right]\left[\begin{array}{cc} \raisebox{-10pt}{$\Wav_{2^{k-1}}$} & \mathbf{1}_{1\times 2^{k-1}} \\
 & 0 \\
-\mathbf{1}_{1\times 2^{k-1}} & \raisebox{-10pt}{$\Wav_{2^{k-1}}$}\\
0 & 
\end{array}\right]\\
&=\left[\begin{array}{cc}
\mathbf{1}_{2^{k-1}\times 2^{k-1}} & 0 \\
0 & \mathbf{1}_{2^{k-1}\times 2^{k-1}}
\end{array}\right] + \left[\begin{array}{cc}
\Wav^t_{2^{k-1}}\Wav_{2^{k-1}} & 0 \\
0 & \Wav^t_{2^{k-1}}\Wav_{2^{k-1}}
\end{array}\right]
\end{align*}
The rest of proof follows a similar process of the proof of correctness of eigenvalues and eigenvecotrs of $\H_{2^k}^t\H_{2^k}$.
\end{proof}

Let $\D_{\H_S}$ be the diagonal matrix whose entries are square roots of eigenvalues in Table~\ref{tab:eigens} and $\P_{\H}$ be the matrix whose row vectors are the normalization of eigenvectors in Table~\ref{tab:eigens}. Then $\D_{\H_S}\P_{\H}$ is an eigendecomposition of $\H_n^t\H_n$. Now let us compute the $L_1$ norm of $i$-th column of $\D_{\H_S}\P_{\H}$. Notice the fact that for each eigenvalue $2^j-1$, $1\leq j\leq k$, there exists exactly one eigenvector $v_j$ in Table~\ref{tab:eigens} which is corresponding to this eigenvalue and has non-zero $i$-th entry. Moreover, since there are $2^j$ entries in $v_j$ that are $\pm 1$ and all other entries in $v_j$ are $0$, the nonzero entries in the normalization of $v_j$ are $\pm 2^{-\frac{j}{2}}$. Since the entries in the normalized eigenvector that correspond to eigenvalue $2^{k+1}-1$ are $\pm 2^{-\frac{k}{2}}$, the $L_1$ norm of $i$-th column is:
\[\sum_{j=1}^k 2^{-\frac{j}{2}}\sqrt{2^j-1} + 2^{-\frac{k}{2}}\sqrt{2^{k+1}-1} = \sum_{j=1}^k \sqrt{1-\frac{1}{2^j}} + \sqrt{2-\frac{1}{2^k}}.\]
Therefore the sensitivity of $\D_{\H_S}\P_{\H}$ is $\sum_{j=1}^k \sqrt{1-\frac{1}{2^j}} + \sqrt{2-\frac{1}{2^k}}$. Consider function $f(k)=k+1-\sqrt{1-\frac{1}{2^k}} + \sqrt{2-\frac{1}{2^k}}$, easy computation will show that
\[f(k+1)-f(k)=1+\sqrt{2-\frac{1}{2^k}}-\sqrt{1-\frac{1}{2^{k+1}}} -\sqrt{2-\frac{1}{2^{k+1}}}>0\]
for any positive integer $k$. Since $f(1)\approx  0.06815>0$, we know $f(k)>0$, which means the sensitivity of $\D_{\H_S}\P_{\H}$ is always smaller than the sensitivity of $\H_n$.

\eat{
\begin{proof}
Similar as the case of $\H_n^t\H_n$, let us prove it by induction on $n$. When $n=1$,
\[\W_1=\left[\begin{array}{cc}2 & 0\\0 & 2\end{array}\right],\]
whose eigenvector $[1, 1]$ and $[1, -1]$ and both of them have eigenvalue $2$. Suppose Table~\ref{tab:eigens} gives eigenvalues and eigenvectors of $\W_{n-1}^t\W_{n-1}$. Notice the fact that
\begin{align*}
&\W_n=\left[\begin{array}{cc} \raisebox{-10pt}{$\W_{n-1}$} & \mathbf{1}_{1\times 2^{n-1}} \\
 & 0 \\
-\mathbf{1}_{1\times 2^{n-1}} & \raisebox{-10pt}{$\W_{n-1}$}\\
0 & 
\end{array}\right]\\
&\W^t_{n}\left[\begin{array}{c} \mathbf{1}_{1\times 2^n}\\ 0\end{array}\right]=\mathbf{1}_{2^n\times 2^n}
\end{align*}
Therefore,
\begin{align*}
\W^t_n\W_n&=\left[\begin{array}{cccc}\multicolumn{2}{c}{\W_{n-1}^t} & -\mathbf{1}_{2^{n-1}\times 1} & 0 \\ \mathbf{1}_{2^{n-1}\times 1} & 0 & \multicolumn{2}{c}{\W_{n-1}^t}\end{array}\right]\left[\begin{array}{cc} \raisebox{-10pt}{$\W_{n-1}$} & \mathbf{1}_{1\times 2^{n-1}} \\
 & 0 \\
-\mathbf{1}_{1\times 2^{n-1}} & \raisebox{-10pt}{$\W_{n-1}$}\\
0 & 
\end{array}\right]\\
&=\left[\begin{array}{cc}
\mathbf{1}_{2^{n-1}\times 2^{n-1}} & 0 \\
0 & \mathbf{1}_{2^{n-1}\times 2^{n-1}}
\end{array}\right] + \left[\begin{array}{cc}
\W^t_{n-1}\W_{n-1} & 0 \\
0 & \W^t_{n-1}\W_{n-1}
\end{array}\right]
\end{align*}
Since the eigenvectors in Table~\ref{tab:eigenvectors_of_wQ} are exactly the same as the eigenvectors in Table~\ref{tab:eigenvectors_of_hQ}, follow a similar process of the proof of correctness of Table~\ref{tab:eigenvectors_of_hQ}, we know Table~\ref{tab:eigenvectors_of_wQ} gives all eigenvalues of $\Wav_n^t\Wav_n$.
\end{proof}
}

Notice the fact that both $\mathbf{1}_{1\times 2^{k}}$ and $[ \mathbf{1}_{1\times 2^{k-1}}, -\mathbf{1}_{1\times 2^{k-1}}]$ are eigenvectors of $\Wav_n^t\Wav_n$ corresponding to eigenvalue $2^k$, we know $[ \mathbf{1}_{1\times 2^{k-1}}, \mathbf{0}_{1\times 2^{k-1}}]$ and $[ \mathbf{0}_{1\times 2^{k-1}}, \mathbf{1}_{1\times 2^{k-1}}]$ are also eigenvectors of $\Wav_n^t\Wav_n$ corresponding to eigenvalue $2^k$.
Let $\D_{\Wav_S}$ be the diagonal matrix whose entries are square roots of eigenvalues in Table~\ref{tab:eigens} and $\P_{\Wav}$ be the matrix whose first $n-2$ row vectors are normalization of eigenvectors in Table~\ref{tab:eigens} and last two row vecotrs are normalization of  $[ \mathbf{1}_{1\times 2^{k-1}}, \mathbf{0}_{1\times 2^{k-1}}]$ and $[ \mathbf{0}_{1\times 2^{k-1}}, \mathbf{1}_{1\times 2^{k-1}}]$. Then $\D_{\Wav_S}\P_{\Wav}$ is an eigendecomposition of $\Wav_n^t\Wav_n$. Similar as $\D_{\H_S}\P_{\H}$, the $L_1$ norm of $i$-th column of $\D_{\Wav_S}\P_{\Wav}$ is:
\[\sum_{j=1}^{k-1} 2^{-\frac{j}{2}}\sqrt{2^j} + 2^{-\frac{k}{2}}\sqrt{2^{k-1}} = k+\sqrt{2}-1.\]
Therefore the sensitivity of $\D_{\Wav_S}\P_{\Wav}$ is $k+\sqrt{2}-1$.

\subsection{Error analysis} \label{app:error}

Based on the eigen-decomposition in the previous section. We now can formally analyze the error of $\H_n$ and $\Wav_n$. 
\Compareerrorwh*
\begin{proof}
Recall that the error for any given query $\w$ and strategy $\A$ is $\frac{2}{\epsilon^2}\Delta_A^2\w^t(\A^t\A)^{-1}\w$. Since $\H_n$ and $\Wav_n$ have the same sensitivity, we need only compare the profile term $\w^t(\A^t\A)^{-1}\w$.  Let $\P_\M\D_\M\P_\M^t$ be the spectral decomposition of $(\A^t\A)^{-1}$. Notice that:
\begin{align*}
\tr(\w^t(\A^t\A)^{-1}\w)=&\tr(\w^t\P_\M\D_\M\P_\M^t\w)\\
=&\tr(\P^t_\M \w\w^t\P_\M\D_\M).
\end{align*}
Since $\H_n$ and $\Wav_n$ have the same eigenvectors, the only difference in error is due to the difference in eigenvalues.  From Table~\ref{tab:eigens} we know ratio between their corresponding eigenvalues is in range $[\frac{1}{2}, 2]$, and that all eigenvalues are positive. Therefore, the ratio between their errors of answering $\w$ is in $[\frac{1}{2}, 2]$.
\end{proof}

\MaximumError*

\begin{proof}
Since $\W_n$ and $\H_n$ are asymptotically equivalent, we can derive the error bounds for either.  We analyze the error of $\W_n$.  Let $n=2^k$, consider the range query $[2^k-\frac{1}{3}(4^{\lfloor\frac{k-1}{2}\rfloor+1}-1), 2^k+\frac{1}{3}(4^{\lfloor\frac{k-1}{2}\rfloor+1}-1)]$.  The error of this query is $\Theta(\log^3 n)$, which follows from algebraic manipulation of Equation~\ref{eq:var}, facilitated by knowing the eigen decomposition of $\ep{\W_n}$. Since Xiao et al.~\cite{xiao2010differential} have already shown that the worst case error of $\W_n$ is $O(\log^3 n)$, we know the maximum error of answering any query in $\Wrang$ is $\Theta(\log^3 n)$.

Moreover, it follows from algebraic manipulation that the error of answering any query $\w$ where the number of non-zero entries is 1 is $O(\log^2n)$. Therefore the error of  any 0-1 query is $O(n\log^2n)$. Consider the query $(0,1,0,1,\ldots, 0,1)$: it can can be shown to have error $\Theta(n\log^2 n)$. Therefore the maximum error of answering any query in $\Wbool$ is $\Theta(n\log^2 n)$.
 
 Recall 
 \[\totalerror{\A}{\W} = \frac{2}{\epsilon^2}\Delta_{\A}\tr(\WW(\A^t\A)^{-1}\WW^t).
\]
Total error of workloads $\Wrang$, $\Wbool$ can be computed by applying the equation above to strategies $\H_n, \W_n$ and $\I_n$.
\end{proof}

\section{The geometry of a strategy} \label{app:geometry}
Finding the optimal strategy for a given workload will require considering both the shape of profile and the sensitivity, as discussed in Section~\ref{sec:optimizing}. We use an example to demonstrate the geometry of the error profile in Sec.~\ref{sec:app:profile}. In Sec.~\ref{sec:app:sensitivity}, we look at the geometry of sensitivity and how the $\Q_\A$ matrix of the decomposition of $\A$ affects the sensitivity of $\A$.

\subsection{The geometry of the error profile}\label{sec:app:profile}

\eat{As discussed in the proof of Theorem~\ref{thm:svd}, for any strategy $\A$, the error profile $\M=\ep{\A}$ is a positive definite matrix, that is, one for which $\w\M\w^t > 0$ whenever $\w \neq 0$.  This matches our intuition since the error for any query must be positive.}
As discussed in Sec.~\ref{sec:analysis}, for any strategy $\A$, the error profile $\M=\ep{\A}$ is a positive definite matrix, that is, one for which $\w\M\w^t > 0$ whenever $\w \neq 0$. The following example gives the geometry of two error profiles.

\begin{example}
\label{ex:ellipse}
Figure~\ref{fig:ellipse} shows the ellipses corresponding to two error profiles 
$\M_1 = [ \begin{smallmatrix}
1 & 0 \\ 0 & 1
\end{smallmatrix} ]	$ 
and 
$\M_2 = [\begin{smallmatrix}
2 & \mbox{-}1.5 \\	\mbox{-}1.5  & 2
\end{smallmatrix}]
$.  
Each point in the x-y plane corresponds to a query $\w = [c_1, c_2]$.  Those points on each ellipse correspond to queries such that $\w \M \w^{t} = 1$.  $\M_1$ is a circle whereas $\M_2$ is a stretched and rotated ellipse.  The figure shows that $\w \M_2 \w^{t} < \w \M_1 \w^{t}$ for queries near and along the line $y = x$.

The profile $\M_1$ has eigenvalues $(1,1)$ and eigenvectors $\PM = [\begin{smallmatrix}
1 & 0\\	0 & 1
\end{smallmatrix}]$, indicating no stretching or rotation.  The profile $\M_2$ has eigenvalues $(7/2, 1/2)$ and its eigenvectors correspond to a $45^{\circ}$ rotation, indicating that the major axis is stretched (by a $\sqrt{7}$ ratio to the minor axis) and rotated to align with $y=x$.

\end{example}

	\begin{figure}[t]
	\centering
	\includegraphics[scale=0.4]{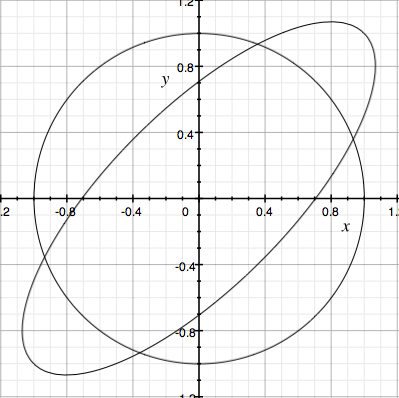}
	\caption{\label{fig:ellipse} For the error profiles $\M_1$ and $\M_2$ described in Example~\ref{ex:ellipse}, this figure shows the ellipses defined by $\w\M_1\w^t=1$, a circle, and $\w\M_2\w^t=1$, an ellipse rotated $45^{\circ}$.  The profile term $f = \w\M_1\w$ is an elliptic paraboloid coming out of the page, centered around the z axis.}
	\end{figure}

	The decomposition can also guide the design of new strategies.  We can design the error profile by choosing values for the diagonal, and choosing a rotation.  Stretching and rotating in a direction makes queries in that direction relatively more accurate than queries in the other directions. 
	
\subsection{The geometry of sensitivity} \label{sec:app:sensitivity}
While we can design a strategy to obtain a desired profile, the error depends not only on the profile, but also on the sensitivity. For example, we can apply Theorem~\ref{thm:profile-strategy} to the previous example.

\begin{example}
	Applying Theorem~\ref{thm:profile-strategy}, we can obtain query strategies $\A_1$ and $\A_2$ that achieve profiles $\M_1$ and $\M_2$ respectively.  %
	$\A_1 = [\begin{smallmatrix}
	1 & 0 \\ 0 & 1
	\end{smallmatrix}]
	$ and %
	$\A_2 = [ \begin{smallmatrix}
	\sqrt{2/7} & 0\\
	0 & \sqrt{2}
	\end{smallmatrix} ]
	[\begin{smallmatrix}
	cos(\pi/4) & -sin(\pi/4)\\
	sin(\pi/4) & cos(\pi/4)
	\end{smallmatrix}]
	= [\begin{smallmatrix}
	1/\sqrt{7} & -1/\sqrt{7}\\
	1 & 1
	\end{smallmatrix}]$.
\end{example}

As it is shown above, $\A_2$ has higher sensitivity than $\A_1$, so while it is more accurate for queries along $y=x$, the difference is less pronounced than Figure~\ref{fig:ellipse} might suggest.  

The sensitivity of a query strategy is determined by its columns.  If $\A$ is decomposed as $\A = \Q_\A\D_\A\P^t_\A$ then the columns of $\P^t_\A$ are orthogonal vectors, but $\D_\A$ stretches the axes so that the columns are no longer necessarily orthogonal. The matrix $\Q_\A$ then rotates the column vectors of $\D_\A\P^t_\A$, but we know that any such rotation will not impact the error profile.  The rotation does impact the sensitivity, because each rotation changes the column vectors and therefore changes the maximum absolute sum of the column vectors.  Since sensitivity is measured by the $L_1$ norm of the column vectors, we can think of an $L_1$ ``ball'' (it is actually diamond shaped) which consists of all points with $L_1$ norm equal to a constant $c$.  If we view the column vectors of $\A$ as points in $n$ dimensional space, the sensitivity is the smallest $L_1$ ball that contains the points.  Minimizing the sensitivity of a given profile (Problem~\ref{problem:minSens}) is therefore equivalent to finding the rotation of the columns in $\A$ that permits them to be contained in the smallest $L_1$ ball.